\documentclass{vldb}
\usepackage{xspace}
\usepackage{graphicx}
\usepackage{balance}  
\usepackage{amsmath}
\usepackage{mathtools}
\usepackage{algorithm}
\usepackage{algorithmic}
\usepackage{subfigure}

\newtheorem{definition}{Definition}
\newtheorem{theorem}{Theorem}
\newtheorem{lemma}{Lemma}

\newtheorem{claim}{Claim}

\newdef{example}{Example}

\newcommand{\eat}[1]{}

\newcommand{\eg}{\textit{\textrm{e.g.}}\xspace}

\newcommand{\ie}{\textit{\textrm{i.e.}}\xspace}
\newcommand{\st}{\textit{\textrm{s.t.}}\xspace}

\newcommand{\etal}{\textit{\textrm{et al.}}\xspace}

\newcommand{\maxethreesat}{\texttt{\small MAX NM-E3SAT}\xspace}
\newcommand{\maxThreeSatBTwentyNine}{\texttt{\small MAX B29-3SAT}\xspace}

\newcommand{\fdset}{$\mathsf{\Sigma}$\xspace}
\newcommand{\fdseta}{$\mathsf{\Sigma_{A\rightarrow{B}\rightarrow{C}}}$\xspace}
\newcommand{\fdsetb}{$\mathsf{\Sigma_{A\rightarrow{B}\leftarrow{C}}}$\xspace}
\newcommand{\fdsetc}{$\mathsf{\Sigma_{AB\rightarrow{C}\rightarrow{B}}}$\xspace}
\newcommand{\fdsetd}{$\mathsf{\Sigma_{AB\leftrightarrow{AC}\leftrightarrow{BC}}}$\xspace}

\newcommand{\NPhard}{\texttt{NPhard}\xspace}

\newcommand{\APXcomplete}{\texttt{APXcomplete}\xspace}

\newcommand{\osrsucc}{{\footnotesize\textsf{OSRSucceed(\fdset)}}\xspace}

\begin{document}


\title{Complexity and Efficient Algorithms for Data Inconsistency Evaluating and Repairing\titlenote{Supported by ~~NSFC xxxx,~~NSFC xxxx}}

\numberofauthors{3} 
\author{
	\alignauthor
	Dongjing Miao\\
	\affaddr{Harbin Institute of Technology}\\
	\affaddr{P.O. Box 321}\\
	\affaddr{92 Xidazhi Street}\\
	\affaddr{Harbin, China}\\
	\email{miaodongjing@hit.edu.cn}    
	\alignauthor
	Zhipeng Cai\\
	\affaddr{Georgia State University}\\
	\affaddr{P.O. Box 5060}\\
	\affaddr{Atlanta, GA, USA}\\
	\email{zcai@gsu.edu}
	\alignauthor
	Jianzhong Li\\
	\affaddr{Harbin Institute of Technology}\\
	\affaddr{P.O. Box 321}\\
	\affaddr{92 Xidazhi Street}\\
	\affaddr{Harbin, China}\\
	\email{lijzh@hit.edu.cn}       
	\and  
	\alignauthor
	Xiangyu Gao\\
	\affaddr{Harbin Institute of Technology}\\
	\affaddr{P.O. Box 321}\\
	\affaddr{92 Xidazhi Street}\\
	\affaddr{Harbin, China}\\
	\email{gaoxy@hit.edu.cn}       
	\alignauthor
	Xianmin Liu\\
	\affaddr{Harbin Institute of Technology}\\
	\affaddr{P.O. Box 321}\\
	\affaddr{92 Xidazhi Street}\\
	\affaddr{Harbin, China}\\
	\email{liuxianmin@hit.edu.cn}
}
\maketitle


\begin{abstract}
Data inconsistency evaluating and repairing are major concerns in data quality management.
As the basic computing task, optimal subset repair is not only applied for cost estimation during the progress of database repairing, but also directly used to derive the evaluation of database inconsistency.
Computing an optimal subset repair is to find a minimum tuple set from an inconsistent database whose remove results in a consistent subset left.
Tight bound on the complexity and efficient algorithms are still unknown.
In this paper, we improve the existing complexity and algorithmic results, together with a fast estimation on the size of optimal subset repair. 
We first strengthen the dichotomy for optimal subset repair computation problem,
we show that it is not only \APXcomplete, but also \NPhard to approximate an optimal subset repair with a factor better than $17/16$ for most cases.
We second show a $(2-0.5^{\tiny\sigma-1})$-approximation whenever given $\sigma$ functional dependencies,
and a $(2-\eta_k+\frac{\eta_k}{k})$-approximation when an $\eta_k$-portion of tuples have the $k$-quasi-Tur$\acute{\text{a}}$n property for some $k>1$.
We finally show a sublinear estimator on the size of optimal \textit{S}-repair for subset queries,
it outputs an estimation of a ratio $2n+\epsilon n$ with a high probability,
thus deriving an estimation of FD-inconsistency degree of a ratio $2+\epsilon$.
To support a variety of subset queries for FD-inconsistency evaluation, we unify them as the $\subseteq$-oracle which can answer membership-query, and return $p$ tuples uniformly sampled whenever given a number $p$.
Experiments are conducted on range queries as an implementation of $\subseteq$-oracle, and results show the efficiency of our FD-inconsistency degree estimator.
\end{abstract}
\section{Introduction}
A database instance $I$ is said to be inconsistent if it violates some given integrity constraints,
that is, $I$ contains conflicts or inconsistencies.
Those database inconsistencies can occur in various scenarios due to many causes.
For example, a typical scenario is information integration,
where data are integrated from different sources,
some of them may be low-quality or imprecise,
so that conflicts or inconsistencies arise.

In the principled approach managing inconsistencies~\cite{arenas99cqa},
the notion of \emph{repair} was first introduced decades ago.
A repair of an inconsistent instance $I$ is a consistent instance $I^\prime$ obtained by performing a minimal set of operations on $I$ so as to satisfy all the given integrity constraints.
Repairs could be defined under different settings of \emph{operations} and \emph{integrity constraints}.
We follow the setting of~\cite{ester:PODS},
where we take functional dependencies, also a most typical one, as the integrity constraints, and deletions as the operations,
so that a repair of $I$ here is a subset of $I$ obtained by minimal tuple deletions,
and an optimal repair of $I$ is a subset of it obtained by deleting minimum tuples.
Computing an optimal subset repair with respect to functional dependencies is the major concern in this paper.
It is a fundamental problem of data inconsistency management and the motivation has been partially discussed in~\cite{ester:PODS}.
The significance of study on computing optimal repair is twofold.

Computing optimal repairs would be the basic task in data cleaning and repairing.
For data repairing, existing methods could be roughly categorized into two classes, fully automatic and semi automatic ways~\cite{10.1145/2882903.2912574}.
In fully automatic repairing methods, optimal subset repairs are always considered as optimization objectives~\cite{10.5555/1325851.1325890,10.1145/3299869.3319901,DBLP:conf/icdt/SaIKRR19}.
Given an inconsistent database, one needs automated methods to make it consistent, \textit{i.e.}, find a repair that satisfies the constraints and minimally differs from the originated input, optimal subset repair is right one of the choices~\cite{ester:PODS}.
On the other side, optimal subset repairs are also preferred candidates picked by automatic data cleaning or repairing system when dealing with inconsistency errors.
Instead of the fully automatic way, the human-in-loop semi-automatic repairing is another prevailing way~\cite{assadi2017dance, bergman2015qoco, dallachiesa2013nadeef, geerts2013llunatic},
and the complement of an optimal subset repair is an ideal lower bound of repairing cost which could be used as to estimate the amount of necessary effort to eliminate all the inconsistency,
sometimes even enlighten them how to choose specific operations.

Besides optimal repairs, measuring inconsistency motivates the computation on the size of optimal repairs.
Intuitively, for the same schema and with the same integrity constraints, given two databases instances, it is natural to know which one is more inconsistent than the other.
This comparison can be accomplished by assigning a measure of inconsistency to a database.
Hence, measuring database inconsistency has been revisited and generalized recently by data management community.
\cite{bertossi:LPNMR} argued that both the admissible repair actions and how close we want stay to the instance at hand should be taken into account when defining such measure.
To achieve this, database repairs~\cite{10.5555/2371212} could be applied to define degrees of inconsistency.
Among a series of numerical measurements proposed in~\cite{bertossi:LPNMR}, subset repair based inconsistency degree $\mathit{inc\text{-}deg}^S$ is the most typical one.
According to~\cite{bertossi:LPNMR}, subset repair based inconsistency degree is defined as the ratio of minimum number of tuple deleted in order to remove all inconsistencies,
\ie, the size of the complement of an optimal subset repair.
Therefore, computing optimal subset repair is right the fundamental of inconsistency degree computation.
Previous studies does not provide fine-grained complexity on this problem and efficient algorithm for large databases.
Thus, we in this paper give a careful analysis on the computational complexity and fast computation on the size of an optimal subset repair.
Contributions of this paper are detailed as follows.

We first study the data complexity of optimal subset repair problem including the lower and upper bounds in order to understand how hard the problem is and how good we could achieve.
The most recent work~\cite{ester:PODS} develops a simplification algorithm $\mathsf{OSRSucceeds}$ and establishes a dichotomy based on it to figure the complexity of this problem.
Simply speaking, they show that, for the space of combinations of  database schemas and functional dependency sets, 
(i) it is polynomial to compute an optimal subset repair, if  the given FD set can be simplified into an empty set;
(ii) the problem is APX-complete, otherwise.

As the computation accuracy of the size of an optimal subset is very crucial to our motivation,
we strengthen the dichotomy in this paper by improving the lower bound into concrete constants.
Specifically, we show that it is \NPhard to obtain a $(17/16-\epsilon)$-optimal subset repair for most input cases, and $(69246103/69246100-\epsilon)$--optimal subset repair for all the others.
We show that a simple reduction could unify most cases and improve the low bound
We then consider approximate repairing. For this long standing problem, it is always treated as a vertex cover problem equivalently, and admits the upper bound of ratio $2$.
However, we  take a step further, show that (i) an $(2-0.5^{\tiny\sigma-1})$-approximation of an optimal subset repair could be obtained for given $\sigma$ functional dependencies, more than that, (ii) it is also polynomial to find an $(2-\eta_k+\frac{\eta_k}{k})$-approximation, which is much better if an $\eta_k$-portion of tuples have the $k$-quasi-Tur$\acute{\text{a}}$n property for some $k>2$.

Then, we turn to the most related problem, to estimate the subset repair based FD-inconsistency degree efficiently.
For an integrated database instance, it is helpful to measure the inconsistency degree of any part of it locally, in order to let users know and understand well the quality of their data.
Consider an inconsistent database $I$ integrated by data from two organizations $A$ and $B$, we need to know the main cause of the conflicts.
If we know the inconsistency degree of some part $A$ is very low but that of $B$ is as high as the inconsistency degree of $I$,
then we could conclude that the cause of inconsistencies is mainly on the conflicts in between, but not in any single source.
That is when we find the inconsistency degree of some local part is approximately equal to that of the global one,
then it is reasonable to take this part as a primary cause of inconsistency, so that we may focus on investigating what happens in $B$.

To this motivation, in this paper, we focus on fast estimating the subset-repair based FD-inconsistency degree.
We here follow the definition of subset repair based measurement recently proposed by Bertossi~\cite{bertossi:LPNMR}, and develop an efficient method estimating FD-inconsistency degree of any part of the input database.
Concretely, it seems a same problem as computing an optimal repair itself, so that the complexity result of optimal subset repair computing indicates it is hard to be approximated within a better ratio polynomially, not to mention linear or even a sublinear running time.
However, we observe that, the value of inconsistency degree is a ratio to the size of input data, say $n$, hence, an $n$-fold accuracy loss of optimal subset repair size estimating is acceptable.
Therefore, we develop a sample-based method to estimate the size of an optimal subset repair with a error of $(2\pm\epsilon)n$ so as to break through the limitation of linear time complexity while achieving an approximation with an additive error $\epsilon$.
To support a variety of subset queries, especially for whose result is very large, we model those queries as the $\subseteq$-oracle which can answer membership-query, and return $k$ tuples uniformly sampled whenever given a number $k$.

The following parts of this paper is organized as follows.
Necessary notations and definitions are formally stated in Section 2.
Complexity results and LP-based approximations are shown in Section 3.
Sampling-based fast FD-inconsistency degree estimation is given in Section 4.
Experiment results are discussed in Section 5.
At last, we conclude our study in Section 6.


\begin{figure*}[!t]
	\begin{small}
		\centering
		\vspace{0mm}
		\subfigure[Example instance \textsf{order}]{
			\hspace{-4mm} \includegraphics[height=35mm]{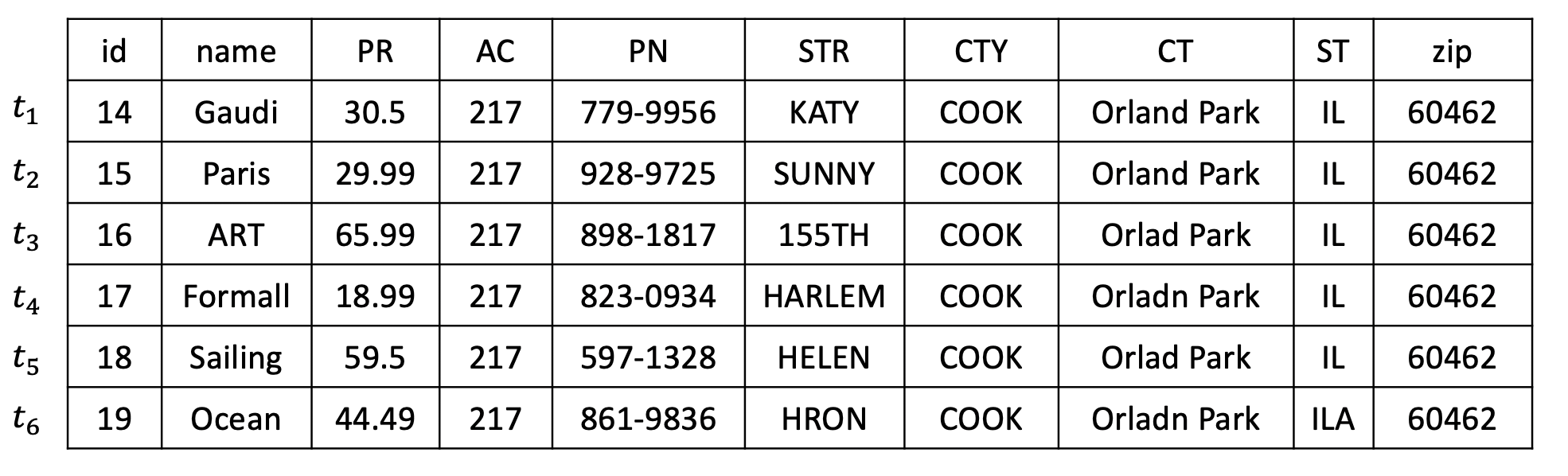}
			\label{fig:eg_ins}
		}
		\subfigure[All conflicts]{
			\includegraphics[height=40mm]{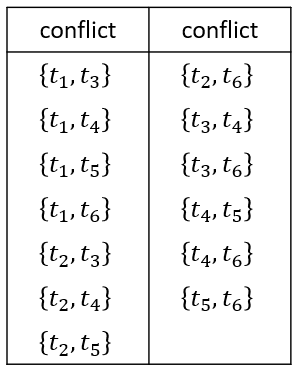}
			\label{fig:eg_cfs}
		}
		\vspace{-3mm}
		\caption{$dist$ with different $\rho$, $n$ and $\sigma$}
		\label{fig:eg}
		\vspace{0mm}
	\end{small}
\end{figure*}

\section{Problem Statement}
The necessary definitions, notations and problem definition are formally given in this section.

\textbf{{Schemas and Tables}}.
A $k$-ary relation schema is represented by $R(A_1,...,A_k)$, where $R$ is the relation name and $A_1,\dots,A_k$ are distinct attributes of $R$.
In the following part of this paper, we refer $R(A_1,...,A_k)$ to $R$ for simplicity. We customarily use capital letters from the beginning of the English alphabet to denote individual attribute, such as ``$A,B,C$'', and use capital letters from the end of the English alphabet individual attribute to denote a set of attributes, such as ``$\mathsf{X},\mathsf{Y},\mathsf{Z}$'', sometimes with subscripts.
A set of attributes are conventionally written without curly braces and commas, such as $\mathsf{X}$ can be written as $AB$ if $\mathsf{X}=\{A,B\}$.\\
We assume the domain of each attribute, $dom(A_i)$, is countably infinite,
then, any instance $I$ over relation $R$ is a collection of $k$-ary tuples $\{a_1,a_2,\dots,a_k\}$,
where each value $a_i$ are taken from the set $dom(A_i)$.
Let $\mathbf{t}.A_i$ refer to the value $a_i$ on attribute $A_i$, and $\mathbf{t}.\mathsf{X}$ refer to the sequence of attribute values $a_1,a_2,\dots,a_i$ when $\mathsf{X}=A_1A_2\dots{A_i}$.
We use $[\mathbf{t}.\mathsf{X}]$ to denote the set of all tuples from $I$ sharing the same value of $\mathsf{X}$.
The size of an instance is the number of tuples in it, denoted as $|I|$.
In this paper, any instance $I$ of a relation schema $R$ is a single table corresponding to $R$. \smallskip

\textbf{{Functional Dependencies}}.
Let $\mathsf{X}$ and $\mathsf{Y}$ be two arbitrary sets of attributes in a relation schema $R$, then $\mathsf{Y}$ is said to be functionally determined by $\mathsf{X}$, written as $\mathsf{X}\rightarrow\mathsf{Y}$, if and only if each $\mathsf{X}$-value in $R$ is associated with precisely one $\mathsf{Y}$-value in $R$.
Usually $\mathsf{X}$ is called the \emph{determinant} set and $\mathsf{Y}$ the \emph{dependent} set, but in this paper, for the sake of simple, we just call them \emph{determinant} and \emph{dependent} respectively.
A functional dependency $\mathsf{X}\rightarrow\mathsf{Y}$ is called \emph{trivial} if $\mathsf{Y}$ is a subset of $\mathsf{X}$.\par
Given a functional dependency $\varphi\colon\mathsf{X}\rightarrow\mathsf{Y}$ over $R$, any instance $I$ corresponding to $R$ is said to satisfy $\varphi$, denoted as $I\models\varphi$, such that for any two tuples $\mathbf{s},\mathbf{t}$ in $I$, $\mathbf{s}.\mathsf{Y}=\mathbf{t}.\mathsf{Y}$ if $\mathbf{s}.\mathsf{X}=\mathbf{t}.\mathsf{X}$.
That is, two tuples sharing the same values of $\mathsf{X}$ will necessarily have the same values of $\mathsf{Y}$.
Otherwise, $I$ does not satisfy $\varphi$, denoted as $I\nvDash\varphi$.
As a special case, any two-tuple subset $J$ of $I$ is called a $\varphi$-\emph{conflict} in $I$ with respect to $\varphi$ if $J\nvDash\varphi$.\par
Let \fdset be a set of functional dependencies, we usually use $\sigma$ to refer to the number of functional dependencies in \fdset, \ie $\sigma=|\mathsf{\Sigma}|$.
Given a set of functional dependency \fdset, an instance $I$ is said to be consistent with respect to \fdset if $I$ satisfies every functional dependencies in \fdset. Otherwise, $I$ is inconsistent, denoted as $I\nvDash\text{\fdset}$.
As a special case, any two-tuple subset $J$ of $I$ is called a \emph{conflict} in $I$ if there is some $\varphi$ such that $J$ is a $\varphi$-\emph{conflict}.
That is, $I$ contains one or more \emph{conflicts} if $I$ is inconsistent.
\begin{example}
\label{eg:sch&fd}
Our running example is around the schema \textsf{order(id, name, AC, PR, PN, STR, CTY, CT, ST, zip)}. Each tuple contains information about an item sold (a unique item \textsf{id}, \textsf{name} and price \textsf{PR}), and the phone number (area code \textsf{AC}, phone number \textsf{PN}) and the address of the customer who purchased the item (street \textsf{STR}, country \textsf{CTY}, city \textsf{CT}, state \textsf{ST}). An instance $I$ of the schema \textnormal{\textsf{order}} is shown in figure \ref{fig:eg_ins}. Some functional dependencies on the order database include:
\begin{equation*}
\begin{aligned}
&\mathsf{fd_1 : [AC, PN] \to [STR, CT, ST]} \quad &\mathsf{fd_2} &\mathsf{: [zip] \to [CT, ST]}\\
&\mathsf{fd_3 : [id] \to [name, PR]} &\mathsf{fd_4} &\mathsf{: [CT, STR] \to [zip]}
\end{aligned}
\end{equation*}
The database of figure \ref{fig:eg_ins} is inconsistent since there are 13 $\mathsf{fd_2}$-conflicts in total as listed in figure \ref{fig:eg_cfs}. The meaning of the number assigned to each conflict will be clarified later.
\end{example}

\textbf{{Equivalence Class}}.
Given an FD $\varphi$: $\mathsf{X}\rightarrow\mathsf{Y}$, an instance $I$ can be partitioned horizontally into several \emph{determinant equivalence classes} according to the $\mathsf{X}$-values, that is, tuples in each determinant equivalence class share the same value of $\mathsf{X}$.
Moreover, any {determinant equivalence class} can be further partitioned into several \emph{determinant-dependent equivalence classes} according to the $\mathsf{Y}$-values, denoted as $[xy]$, that is, tuples in each determinant equivalence class share the same value of $\mathsf{XY}$.
It is obviously that, for an FD $\varphi:\mathsf{X}\rightarrow\mathsf{Y}$ and any instance $I$, two tuples $\mathbf{s}$ and $\mathbf{t}$ not in any $\varphi$-conflict in $I$ must be in different {determinant equivalence classes} $[\mathbf{s}.\mathsf{X}]$ and $[\mathbf{t}.\mathsf{X}]$ respectively. 
\begin{example}
\label{eg:dec}
With respect to $\mathsf{fd_2 : [zip] \to [CT, ST]}$, instance $I$ can be partitioned into one determinant equivalence class $[\mathsf{60462}] = \{\mathbf{t}_1, \mathbf{t}_2, \mathbf{t}_3, \mathbf{t}_4, \mathbf{t}_5, \mathbf{t}_6\}$
and 4 determinant-dependent equivalence classes $\{\mathbf{t}_1, \mathbf{t}_2\}$, $\{\mathbf{t}_3, \mathbf{t}_5\}$, $\{\mathbf{t}_4\}$ and $\{\mathbf{t}_6\}$.
\end{example}

\textbf{{Repair}}.
Let $I$ be an instance over a relation schema $R$,
a subset of $I$ is an instance $J$ obtained from $I$ by eliminating some tuples.
If $J$ is a subset of $I$, then the distance from $J$ to $I$, denoted $dist_{\texttt{sub}}(J,I)$,
is the number of tuples missing from $I$, and it is for sure that $J\subseteq{I}$, thus,
\[dist_{\texttt{sub}}(J,I)=|I\setminus{J}|=|I|-|J|\]\quad
Let $I$ be an instance over schema $R$, and let \fdset be a set of FDs.
A consistent subset of $I$ with respect to \fdset is a subset $J$ of $I$ such that $J\models\text{\fdset}$.
A subset repair (\emph{s-repair}, for short) is a consistent subset that is not strictly contained in any other consistent subset.
An optimal subset repair of $I$ is a consistent subset $J$ of $I$ such that $dist_{\texttt{sub}}(J,I)$ is minimum among all consistent subsets of $I$.
Note that, each optimal subset repair is a repair, but not necessarily vice versa.
Clearly, any consistent subset can be polynomially transformed into a subset repair, with no increase of distance.
Unless explicitly stated otherwise, in this paper, we do not distinguish between a {subset repair} and a {consistent subset}.
\begin{example}
\label{eg:osr}
Both $S_1 = \{\mathbf{t}_4\}$ and $S_2 = \{\mathbf{t}_1, \mathbf{t}_2\}$ are s-repairs of $I$. It is easy to verify that $S_2$ is an optimal s-repair such that $dist_{\texttt{sub}}(S_2,I) = 4$.
\end{example}

Now, we formally define the first problem studies in this paper as follows,
\begin{definition}[OSR Computing]
	Input an instance $I$ over a relation schema $R$, a functional dependency set \fdset, OSR computing problem is to compute an optimal s-repair $J$ of $I$ with respect to \fdset.
\end{definition}

\textbf{Inconsistency Measurement}.
Computing an optimal s-repair helps estimating database FD-inconsistency degree.
As in literature~\cite{bertossi:LPNMR}, given a functional dependency set \fdset, one of subset repair based measurements on the FD-inconsistency degree of input database $I$ is defined as following,
\[incDeg(I,\mathsf{\Sigma}) = \min\limits_{\substack{J\subseteq{I},\\J\models{\mathsf{\Sigma}}}} \left\lbrace\frac{dist_{\texttt{sub}}(J,I)}{|I|}\right\rbrace = \frac{dist_{\texttt{sub}}(J_{opt},I)}{|I|}\]
Moreover, this measurement could be also applied for any part $H$ of the input database $I$ in order to evaluate its corresponding FD-inconsistency degree as following,
\[incDeg(I,H,\mathsf{\Sigma}) = \min\limits_{\substack{J\subseteq{H},\\J\models{\mathsf{\Sigma}}}} \left\lbrace\frac{dist_{sub}(J,H)}{|H|}\right\rbrace=incDeg(H,\mathsf{\Sigma})\]
The local degree does not depends on the whole of the input data, thus leads to the right equation.
Our FD-inconsistency degree of any part is defined locally, hence, we use notation $incDeg(H,\mathsf{\Sigma})$ instead of $incDeg(I,H,\mathsf{\Sigma})$ by omitting the first parameter.
Then, we here formally define the second problem studied in this paper as follows, 
\begin{definition}[FD-inconsistency Evaluation]
	Input a relation schema $R$, an FD set \fdset, an instance $I$ over $R$ and a subset query $Q$ on $I$, FD-inconsistency evaluation is to compute $incDeg\left(Q(I),\mathsf{\Sigma}\right)$ of the query result $Q(I)$ with respect to \fdset.
\end{definition}
\begin{example}
\label{eg:incDeg}
As mentioned in Example \ref{eg:osr}, $S_2 = \{\mathbf{t}_1, \mathbf{t}_2\}$ is an optimal s-repair of $I$, then $incDeg(I,\mathsf{\Sigma}) = \frac{dist_{\texttt{sub}}(S_2,I)}{|I|} = \frac{2}{3}$. Given a range query $Q = [15, 45]$ on attribute \textsf{PR} in \textsf{order}, the result set $Q(I) = \{\mathbf{t}_1, \mathbf{t}_2, \mathbf{t}_4, \mathbf{t}_6\}$. $S_2$ ia also an optimal s-repair of $Q(I)$, then $incDeg(Q(I),\mathsf{\Sigma}) = \frac{dist_{\texttt{sub}}(S_2,I)}{|Q(I)|} = \frac{1}{2}$.
\end{example}

\textbf{Approximation}.
We follow the convention of approximation definition, to define the approximation of optimal repairs explicitly.
For a constant $c\geq1$, a $c$-optimal \emph{s-repair} is an \emph{s-repair} $J$ of $I$ such that\[dist_{\texttt{sub}}(J,I) \leq c\cdot dist_{\texttt{sub}}(J^\prime,I)\]for all \emph{s-repairs} $J^\prime$ of $I$.
In particular, an optimal \emph{s-repair} is the same as a $1$-optimal \emph{s-repair}.

According to the definition of subset repair based FD-inconsistency degree,
for an arbitrary $0\leq\epsilon\leq1$ and a constant $c\geq1$, $\tilde{incDeg}(I,\mathsf{\Sigma})$ is
a $(c,\epsilon)$-approximation of $incDeg(I,\mathsf{\Sigma})$ such that
\[incDeg(I,\mathsf{\Sigma}) \leq \tilde{incDeg}(I,\mathsf{\Sigma}) \leq c\cdot{incDeg(I,\mathsf{\Sigma})} + \epsilon\]

\textbf{Complexity}.
The conventional measure of \emph{data complexity} are adopted to perform the computational complexity analysis of optimal subset repair computing problem in this paper.
That is, the relation schema $R(A_1,\dots,A_k)$ and the functional dependency set \fdset are fixed in advance, and the instance data $I$ over $R$ is the only input. Therefore, an polynomial running time may have an exponential dependency on $k$ and $|\mathsf{\Sigma}|$.
In such context of data complexity, each distinct setting of $R(A_1,\dots,A_k)$ and \fdset indicates a distinct problem of finding an optimal repair, so that different setting may indicate different complexities.
Recall that, in the measurement of combined complexity, the relation schema and the functional dependency set are considered as inputs, hence, the hardness of OSR computing problem equals to that of vertex cover problem. However, this is not the case under data complexity.\par
After showing the hardness, we still adopt data complexity to be the measurement on running times and approximation ratios, however, the difference is that we fix only the size of the functional dependency set \fdset, but not itself and the schema.
Note that, this is reasonable in practical, the input functional dependencies may vary with time, but the number of given functional dependencies are always much smaller than the size of input data $I$, so that we could consider it to be bounded within some constant.\par\smallskip

\section{Computing An Optimal S-Repair}
In this section, we show the improved lower bound and upper bound of OSR.
\subsection{The Strengthened Dichotomy for OSR}
Livshits \etal gave a procedure \osrsucc~\cite{ester:PODS} to simplify a given functional dependency set \fdset.
Any functional dependency set can either be simplified polynomially into a set containing only \emph{trivial} functional dependencies, or not.
The procedure \osrsucc returns \texttt{\textit{true}} for the former case, otherwise \texttt{\textit{false}}.
OSR is polynomially tractable for functional dependency sets that can be simplified into trivial ones.
For all the other functional dependency sets, OSR computing problem is hard as in not only \NPhard but also \APXcomplete.\par
Specifically, any functional dependency set that cannot be simplified further can be classified into one of five certain classes of functional dependency sets.
And OSR is shown in \APXcomplete for any such functional dependency set by \emph{fact}-\emph{wise} reductions from one of the following four fixed schemas.\par
\begin{tabular}{lll}
	\fdseta & = & $\{A\rightarrow{B},B\rightarrow{C}\}$\\
	\fdsetb & = & $\{A\rightarrow{B},C\rightarrow{B}\}$\\
	\fdsetc & = & $\{AB\rightarrow{C},C\rightarrow{B}\}$\\
	\fdsetd & = & $\{AB\rightarrow{C},AC\rightarrow{B},BC\rightarrow{A}\}$
\end{tabular}\smallskip

By showing the inapproximability of such four schemas, the following dichotomy follows immediately.
\begin{theorem}[Dichotomy for OSR computing~\cite{ester:PODS}]
Let \fdset be a set of FDs, then
\begin{itemize}
\item An optimal subset repair can be computed polynomially, if \emph{\osrsucc} returns \texttt{true};
\item Computing an optimal subset repair is \APXcomplete, if \emph{\osrsucc} returns \texttt{false}.
\end{itemize}
\label{thm:d1}
\end{theorem}

In this paper, we give a more careful analysis to show a concrete constant for each of the four schemas, thus strengthening this dichotomy.
\begin{lemma}
For FD sets~\fdseta,~\fdsetb,~and~\fdsetc, there is no polynomial-time $(\frac{17}{16}-\epsilon)$-approximation algorithm for computing an optimal subset repair for any $\epsilon>0$, unless \emph{\texttt{NP=P}}.
\label{lem:1.068}
\end{lemma}
\begin{proof}
We here show three similar gap-preserved reductions (\ie, $\prec_\mathcal{G}$) from \maxethreesat to show the lower bound.
Note that, 
(i) any variable $x$ do not occur more than once in any clause, and
(ii) each clause is \textit{monotone} as either ($x$+$y$+$z$) or ($\bar{x}$+$\bar{y}$+$\bar{z}$).
Each of the following reductions generates a corresponding table instance $I$ with schema $\left\langle{A,B,C}\right\rangle$.\smallskip

{\small \textbf{\maxethreesat$\prec_\mathcal{G}$ \fdseta}}.
For each clause $c_i$, build three tuples.
If $c_i$ contains a positive literal of variable $x_j$, then build $(i,j,j)$.
If $c_i$ contains a negative literal of variable $x_j$, then build $(i,j,\bar{j})$.
Intuitively, $A\rightarrow{B}$ guarantees that exactly one of the three tuples survives once the corresponding clause is satisfied, and $B\rightarrow{C}$ will ensure the consistent assignment of each variable.\smallskip
	
{\small \textbf{\maxethreesat$\prec_\mathcal{G}$ \fdsetb}}.
By simply exchange the column B and C, we get the second reduction.
Concretely, for each clause $c_i$, build three tuples.
If $c_i$ contains a positive literal of variable $x_j$, then build $(i,j,j)$.
If $c_i$ contains a negative literal of variable $x_j$, then build $(i,\bar{j},j)$.
Intuitively, $A\rightarrow{B}$ guarantees that exactly one of the three tuples survives once the corresponding clause is satisfied, and $C\rightarrow{B}$ will ensure the consistent assignment of each variable.\smallskip

{\small \textbf{\maxethreesat$\prec_\mathcal{G}$ \fdsetc}}.
By slightly modify the way of tuple generation, we get the third reduction.\par
Concretely,
(i) for each variable $x_i$, build two $x$-tuples $(x_i,1,x_i)$ and $(x_i,0,x_i)$,
(ii) for each clause $c_i$, build three $c$-tuples, if $c_i$ contains a positive literal of variable $x_j$, then build $(c_i,1,x_j)$,if $c_i$ contains a negative literal of variable $x_j$, then build $(c_i,0,x_j)$.

Intuitively, there are $2n+3m$ tuples created, $AB\rightarrow{C}$ guarantees that exactly one of the three tuples survives once the corresponding clause is satisfied, and $C\rightarrow{B}$ will ensure the consistent assignment of each variable.\smallskip
	
\textbf{\texttt{Lower bound}}.
We here show the proof for \fdsetc, and the other two are similar.
For any instance $\phi$ of \maxethreesat problem, we denote the corresponding table built by reduction as $I_\phi$.
Let $\#\tau(\phi)$ be the number of clauses satisfied by an assignment $\tau(\phi)$ on $\phi$,
and $\#\tau_{\max}(\phi)$ be the number of clauses satisfied by an optimal assignment $\tau_{\max}(\phi)$ on $\phi$.
\begin{claim}
Any tuple deletion $\Delta$ should contain at least two of the three tuples having the same value $i$ on the attribute $A$ for any $1\leq i\leq m$.
\end{claim}
\begin{claim}
 any tuple deletion $\Delta$ should contain either the set of tuples $(c_i,1,x_j)$ or the set of tuples $(c_i,0,x_j)$ for any $1\leq i\leq m, 1\leq j\leq n$.
\end{claim}
FD $AB\rightarrow{C}$ guarantees the first claim, and FD $C\rightarrow{B}$ ensures that there is always an assignment $\tau$ can be derived from $I\setminus\Delta$,
\begin{equation*}
\st\qquad\tau\left(x_i\right) = \left\{
\begin{array}{ll}
{0}, &\text{if}\quad{(x_i,1,x_i)\in{I\setminus\Delta}},\\
{1}, &\text{otherwise}.
\end{array}
\right.
\end{equation*}
\begin{claim}
Let $\Delta_{\min}$ be an minimum tuple deletion, then any optimal solution $\Delta_{\min}$ does not contain $(x_i,1,x_i)$ and $(x_i,0,x_i)$ simultaneously for each variable $x_i$.
\end{claim}
Proof by contradiction. Suppose if not, there must exist another solution $\Delta^\prime$ obtained by returning tuple $(x_i,1,x_i)$ or $(x_i,0,x_i)$ from $\Delta_{\min}$ into $I\setminus{\Delta_{\min}}$ without producing any inconsistency, thus resulting in a solution $\Delta^\prime$ smaller than the optimal one. Based on the three claims, we have
\[\#\tau_{\max}(\phi) = |I_\phi| - |\Delta_{\min}|- n\]
and for any solution $\Delta$ of $I_\phi$,
\[\#\tau(\phi) \geq |I_\phi| - |\Delta|- n\]
additionally, we have the fact that \[|I_\phi|=2n+3m\]
Now, suppose $\Delta$ is an $r$-approximation ($r>1$) of $\Delta_{\min}$ such that $|\Delta|\leq r\cdot|\Delta_{\min}|$, then
\begin{eqnarray}
\frac{\#\tau{(\phi)}}{\#\tau_{\max}(\phi)}
&\geq&
\frac{|I_\phi|-|\Delta|-n}{|I_\phi|-|\Delta_{\min}|-n}\nonumber\\
&\geq&
\frac{|I_\phi|-r\cdot|\Delta_{\min}|-n}{|I_\phi|-|\Delta_{\min}|-n}\nonumber\\
&=&
1+\frac{\left(1-r\right)\cdot|\Delta_{\min}|}{|I_\phi|-|\Delta_{\min}|-n}
\label{eqn:e1}
\end{eqnarray}
since each clause has exactly 3 literals, we have \[|\Delta_{\min}|\geq n + 2\cdot\frac{|I_\phi|-2n}{3}\]
apply this fact in the right hand of inequality~(\ref{eqn:e1}),
it is\[\frac{|\Delta_{\min}|}{|I_\phi|-|\Delta_{\min}|-n}\geq\frac{2|I_\phi|-n}{|I_\phi|-2n}=2+\frac{3}{\frac{|I_\phi|}{n}-2}\]
since $|I_\phi|=2n+3m>2n$, therefore we get \[\frac{|\Delta_{\min}|}{|I_\phi|-|\Delta_{\min}|-n}>2\]
apply this into inequality~(\ref{eqn:e1}), then \[\frac{\#\tau(\phi)}{\#\tau_{\max}(\phi)}>3-2r\]

That is, if there is an $r$-approximation of OSR, then \maxethreesat can be approximated within $3-2r$.
However, if \texttt{\small OSR} can be polynomially approximated within $\frac{17}{16}$, then there exists a polynomial approximation better than $\frac{7}{8}$ for \maxethreesat problem, but it contraries to the hardness result shown in~\cite{guruswami:CCC}.

One can verify the lower bound of \fdseta and \fdsetb in the same way, then the lemma follows immediately. One can refer to appendix for more detail.
\end{proof}

To deal with the last case, by carefully merging the four $\mathcal{L}_{\alpha,\beta}$-reductions $\{$
\begin{tabular}{rcl}
	\maxThreeSatBTwentyNine & $\prec_{\mathcal{L}_{529,1}}$ & \texttt{3DM} ~\cite{kann1991maximum}\\
	\texttt{3DM} & $\prec_{\mathcal{L}_{1,1}}$ & \texttt{MAX 3SC} ~\cite{kann1991maximum}\\
	\texttt{MAX 3SC} & $\prec_{\mathcal{L}_{55,1}}$ & \texttt{Triangle} ~\cite{amini2009hardness}\\
	\texttt{Triangle} & $\prec_{\mathcal{L}_{\frac{7}{6},1}}$ & \fdsetd~\cite{ester:PODS}
\end{tabular}
$\}$, we have the following lemma.
\begin{lemma}
For FD set~\fdsetd, there is no polynomial time $(\frac{69246103}{69246100} - \epsilon)$-approximation algorithm for computing an optimal subset repair for any $\epsilon>0$, unless \emph{\texttt{NP=P}}.
\label{lem:1.000000048}
\end{lemma}
\begin{proof}
By merging the $\mathcal{L}_{\alpha, \beta}$-reduction mentioned above, if computing an OSR for FD set~\fdsetd can be approximated within $\frac{69246103}{69246100}$, then there exists a polynomial approximation better than $\frac{680}{679}$ for \maxThreeSatBTwentyNine problem which is contrary to the hardness result shown in \cite{crescenzi1997short}.
\end{proof}
Based on Lemma~\ref{lem:1.068},\ref{lem:1.000000048} and Theorem~\ref{thm:d1}, a strengthened dichotomy for OSR computing can be stated as follows.
\begin{theorem}[A strengthened dichotomy for OSR]
Let \fdset be a set of FDs, then
\begin{itemize}
\item An optimal subset repair can be computed polynomially, if \emph{\osrsucc} returns \texttt{true};
\item There is no poly-time ($\frac{69246103}{69246100}-\epsilon$)-approximation to compute an optimal subset repair, if \emph{\osrsucc} returns \texttt{false} and $\Sigma$ can be classified into the class having a fact-wise reduction from \fdsetd to itself;
\item There is no poly-time ($\frac{17}{16}-\epsilon$)-approximation to compute an optimal subset repair, otherwise.
\end{itemize}
\label{thm:d2}
\end{theorem}

For the polynomial-intractable side, one can simply verify that if the size of FD set is unbounded, then the OSR computing is as hard as classical vertex cover problem on general inputs which is \NPhard to be approximate within $2-\epsilon$ for any $\epsilon>0$.
A simple approximation algorithm can provide a ratio of 2 when the input FD set is unbounded.

However, in practical, the size of FD set is usually much smaller than the size of data, so that it can be treated as fixed, especially in the context of big data.
Unfortunately, it is still unclear how good we could arrive when the size of FD set is bounded. Therefore, to study the upper bound of its data complexity, we next give a carefully designed approximation to archive a ratio of $2-0.5^{\tiny\sigma-1}$ when the number of given FDs is $\sigma$, or even better sometimes.

\subsection{Approximation}
To investigate the upper bound of optimal s-repair computing problem, we start from a basic linear programming to provide a ratio of $2-0.5^{\tiny\mathcal{X}(n)}$, for an input instance $I$ over a given relation schema $R$ and an input FD set \fdset, where $\mathcal{X}(I)$ is the number of all possible \emph{determinant-dependent equivalence classes} of an input instance $I$ with respect to the input FD set \fdset.
Then, an improved the approximation ratio $2-0.5^{\tiny\sigma-1}$ could be derived by means of triad elimination.
Finally, we find another $(2-\eta_k+\frac{\eta_k}{k})$-approximation which is sometimes, but not always, better than $2-0.5^{\tiny\sigma-1}$, based on a $k$-quasi-Tur$\acute{\text{a}}$n characterization of the input inconsistent instance with respect to the input \fdset.

\subsubsection{A basic approximation algorithm}
We start from the basic linear programming model which is equivalent to the classical one solving  minimum vertex cover problem.\par
Let $x_i$ be a 0-1 variable indicating the elimination of tuple $\mathbf{t}_i$ such that, $x_i = 1$ if eliminate $\mathbf{t}_i$; $x_i = 0$ otherwise.
Then we formulate the OSR computing problem as followings,
\begin{eqnarray}
minimize & \sum_{\mathbf{t}_i\in I} x_i&\\
s.~t. &  x_i+x_j\geq 1,& \forall~\{\mathbf{t}_i,\mathbf{t}_j\}\nvDash\mathsf{\Sigma},\\
&x_i\geq 0, &\forall~\mathbf{t}_i\in I
\end{eqnarray}
It is well-known that every extreme point of this model takes value of 0 or 0.5 or 1, hence, we can relax it with condition: \[x_i\in\{0,0.5,1\}\] thus getting \[OPT^{relax}\leq{OPT}\]
A trivial rounding derives a ratio of $2$ immediately.
However, based on a partition of instance $I$ with respect to FD set \fdset, a better ratio depending on the size of partition could be obtained.

Obviously, for any FD $\varphi_i:\mathsf{X}_i\rightarrow\mathsf{Y}_i$ of \fdset with a size of $\sigma$, each tuple $\mathbf{t}$ belongs to one and only one distinct {determinant-dependent equivalence class} with respect to $\varphi_i$, say $[\mathbf{t}.\mathsf{X}_i\mathsf{Y}_i]$, then we have \[\mathbf{t}\in[\mathbf{t}.\mathsf{X}_1\mathsf{Y}_1]\cap\dots\cap[\mathbf{t}.\mathsf{X}_\sigma\mathsf{Y}_\sigma]=[\mathbf{t}.\mathsf{Z}],\]
where $\mathsf{Z}=\mathsf{X}_1\cup\dots\cup\mathsf{X}_\sigma\cup\mathsf{Y}_1\cup\dots\cup\mathsf{Y}_\sigma.$
Hence, we observe that if any two tuples $\mathbf{s}$ and $\mathbf{t}$ are in some conflict, then there must be
\[\mathbf{s}\notin[\mathbf{t}.\mathsf{Z}],\mathbf{t}\notin[\mathbf{s}.\mathsf{Z}],\]
and vice versa, since they disagree on at least one attribute in some $\mathsf{Y}_i$ but agree on all the attributes in $\mathsf{X}$.

Further more, another observation is that all the tuples in conflict with $\mathbf{t}$ are included in the determinant equivalence classes
$$[\mathbf{t}.\mathsf{X}]=[\mathbf{t}.\mathsf{X}_1]\cup\dots\cup[\mathbf{t}.\mathsf{X}_\sigma]$$
Because all tuples in each $[\mathbf{t}.\mathsf{X}_i]$ may be inconsistent with each other at worst, hence,
every tuple in $[\mathbf{t}.\mathsf{X}]$ may be inconsistent with at most $|[\mathbf{t}.\mathsf{X}_1]|\times\dots\times|[\mathbf{t}.\mathsf{X}_\sigma]|-1$ tuples.

Let $\mathcal{X}(\mathbf{t})$ be the numbers of tuples who are in conflict with $\mathbf{t}$,
and $\mathcal{X}(I)$ be the numbers of consistent classes that $I$ could be partitioned into,
such that each class is consistent.
This observation implies the following claims immediately,
\begin{claim}
\small$\mathcal{X}(I)\leq\max\limits_{\mathbf{t}\in{I}}\{\mathcal{X}(\mathbf{t})\}\leq\max\limits_{\mathbf{t}\in{I}}\{|[\mathbf{t}.\mathsf{X}_1]|\times\dots\times|[\mathbf{t}.\mathsf{X}_\sigma]|\}$
\end{claim}
This claim implies that all the tuples in $I$ could be partitioned into at most $\mathcal{X}(I)$ classes such that tuples in each class are consistent with each other.

Then, we improve the ratio by using the $\mathcal{X}(I)$ partitions of the input instance.
Based on the rounding technique similar with~\cite{Nemhauser1975}, an improved approximated algorithm could be stated as follows.
\begin{algorithm}
	\caption{\textsf{Baseline LP-OSR}}  
	\hspace*{0.02in}{\bf Input:}
	\quad $n$-tuple instance $I$ over schema $R$, FD set \fdset\\
	\hspace*{0.02in}{\bf Output:} 
	optimal subset repair $J$ of $I$ with respect to \fdset\\
	\begin{algorithmic}[1]
		\STATE{Solve the linear programming (2)-(4) to obtain a solution $x_1\dots x_n$ such that $x_i\in\left\lbrace0,0.5,1\right\rbrace$ for all $1\leq i\leq n$.}
		\STATE{Let $P_j$ is the set of tuples of some consistent partition of $I$ with respect to \fdset}
		\STATE{$j\leftarrow \arg\max_j \left| \left\lbrace x_i|\hspace{0.5ex} \mathbf{t}_i\in P_j \wedge x_i=0.5\right\rbrace\right|$}
		\FOR{each $\mathbf{t}_i\in I$}
		\IF{$x_i=1$ or ($x_i=0.5$ and $\mathbf{t}_i\notin{P_j}$)}
		\STATE{add $\mathbf{t}_i$ into $\hat{\Delta}$}
		\ENDIF
		\ENDFOR
		\STATE{$\hat{J}\rightarrow I\setminus\hat{\Delta}$}
		\RETURN{$\hat{J}$}
	\end{algorithmic}
\label{algo:baseline-lp-osr} 
\end{algorithm}

Obviously, $OPT^{relax}$ can be returned in polynomial time as shown in~\cite{williamson2011design}, and it is easy to see that $\hat{J}$ is a s-repair.
In fact, if an $x_i=0.5$ and not be picked into deletion $\hat{\Delta}$, then all the tuples in conflicts with $\mathbf{t}_i$ must be added into $\hat{\Delta}$
because they are not in partition $P_j$, and the sum of two variables of tuples in any conflict should be no less than $1$.
Therefore, we claim that the approximation ratio is $2-\frac{2}{\mathcal{X}(I)}$.
\begin{lemma}
	Algorithm~\ref{algo:baseline-lp-osr} returns a $\left(2-\frac{2}{\mathcal{X}(I)}\right)$-optimal subset repair.
	\label{lem:2-k}
\end{lemma}
\begin{proof}
	Define notations $S_1$, $S_{0.5}$ and $S_{P_j}$ as follows,
	\begin{eqnarray*}
	S_1&=&\left\lbrace x_i|~\mathbf{t}_i\in I, x_i=1\right\rbrace\\
	S_{0.5}&=&\left\lbrace x_i|~\mathbf{t}_i\in I, x_i=0.5\right\rbrace\\
	S_{P_j}&=&\left\lbrace x_i|~\mathbf{t}_i\in P_j, x_i=0.5\right\rbrace
	\end{eqnarray*}
	Then, the following holds obviously,
	\[OPT^{relax}\coloneqq\sum_{\mathbf{t}_i\in S_{1}\cup S_{0.5}} x_i \leq OPT \leq |\Delta_{min}| = dist_{\texttt{sub}}(J_{opt},I)\]
	Second, We have that
	\begin{eqnarray*}
	dist_{\texttt{sub}}(\hat{J},I)&=&|\hat{\Delta}|\\
	&\leq&
	|S_1| + |S_{0.5}| - |S_{P_j}|\\
	&=&
	\sum_{i\in S_{1}} x_i +
	2\sum_{i\in S_{0.5}} x_i- 2\sum_{i\in S_{P_j}} x_i\\
	&\leq&
	\sum_{i\in S_{1}} x_i +
	2\sum_{i\in S_{0.5}} x_i- 2\cdot\frac{1}{{\mathcal{X}(I)}}\sum_{i\in S_{0.5}} x_i\\
	&\leq&
	\sum_{i\in S_{1}} x_i + \left(2-\frac{2}{\mathcal{X}(I)}\right)\sum_{i\in S_{0.5}} x_i\\
	&\leq&
	\left(2-\frac{2}{\mathcal{X}(I)}\right)\sum_{i\in S_{1}\cup S_{0.5}} x_i\\
	&\leq&
	\left(2-\frac{2}{\mathcal{X}(I)}\right)\cdot OPT\\
	&\leq&
	\left(2-\frac{2}{\mathcal{X}(I)}\right)\cdot dist_{\texttt{sub}}(J_{opt},I)
	\label{eqn:2-1C}
	\end{eqnarray*}
	Therefore, $\hat{J}$ is a $\left(2-\frac{2}{\mathcal{X}(I)}\right)$-optimal subset repair of $I$.
\end{proof}
The number $\mathcal{X}(I)$ is unbounded, in the worst case, could be as large as $|I|$ so that it is a factor depending on the size of input.
\subsubsection{Improved ratio by triad eliminating}
Reducing the number of consistent partitions will improve the approximation.
We introduce triad elimination in this section to decrease the partition number into a factor which is independent with the size of input but only depending on the number $\sigma=|\mathsf{\Sigma}|$ of input functional dependencies.

\emph{Data reduction}.
Let $\mathbf{r},\mathbf{s},\mathbf{t}$ be three tuples in $I$, then they are called a \emph{triad} if any two of them are in a conflict with respect to \fdset.
An important observation is that any s-repair contains at most one tuple of a triad in $I$, especially in an optimal s-repair,
hence, any triad elimination yields a $1.5$-optimal s-repair of itself.
Therefore, we could preform a data reduction by eliminating all the disjoint triads without the loss of an approximation ratio $1.5$.

Based on the data reduction, the improved algorithm can be shown as follow.
\begin{algorithm}
	\caption{\textsf{TE LP-OSR}}  
	\hspace*{0.02in}{\bf Input:}
	\quad $n$-tuple instance $I$ over schema $R$, FD set \fdset\\
	\hspace*{0.02in}{\bf Output:} 
	optimal subset repair $J$ of $I$ with respect to \fdset\\
	\begin{algorithmic}[1]
		\STATE{Find a maximal tuple set of disjoint triads $\Delta$ from $I$}
		\STATE{$I^\prime\leftarrow{I\setminus{\Delta}}$}
		\STATE{$\hat{J}\leftarrow\text{\textsf{Baseline LP-OSR}}(I^\prime)$}
		\RETURN{$\hat{J}$}
	\end{algorithmic}
	\label{algo:te-lp-osr} 
\end{algorithm}

Let $\sigma$ be the number of functional dependencies in \fdset, then we claim that \textsf{TE LP-OSR} will return a better approximation as the following theorem.
\begin{theorem}
Algorithm~\textsf{\emph{TE LP-OSR}} returns a $\left(2-0.5^{\sigma-1}\right)$-optimal subset repair.
\label{thm:2-0.5}
\end{theorem}
\begin{proof}
Algorithm~\ref{algo:te-lp-osr} does find an s-repair, because all the triads are eliminated from $I$, hence conflicts involving any tuples in $\Delta$ are removed from $I$,
and all conflicts in the reduced data $I\setminus{\Delta}$ are removed by \textsf{Baseline-LP-OSR}.\par
Moreover, any optimal s-repair of $I$ contains at most one tuples of a triad which yields a $1.5$-approximation for the subset $\Delta$ of $I$, formally, we have,
\[dist_{\texttt{sub}}(\emptyset,\Delta) \leq 1.5\cdot dist_{\texttt{sub}}(\Delta_{opt},\Delta)\]
Additionally, consider each \emph{determinant equivalence class} with respect to any single FD,
no triad in it, that is, $I$ could be partitioned into 2 consistent classes with respect to this FD.
It implies $I$ could be partitioned into $2^\sigma$ consistent classes with respect to \fdset.
Due to lemma~\ref{lem:2-k}, $J^\prime$ is a $\left(2-0.5^{\sigma-1}\right)$-optimal s-repair of $I\setminus{\Delta}$.\par\smallskip
\noindent Let $J^\prime_{opt}$ be the optimal s-repair of $I\setminus{\Delta}$, then,
\[dist_{\texttt{sub}}(\hat{J},I\setminus{\Delta})\leq\left(2-0.5^{\sigma-1}\right)\cdot dist_{\texttt{sub}}(J^\prime_{opt},I\setminus{\Delta})\]
And $\hat{J}\cup\emptyset$ is an s-repair of $(I\setminus\Delta)\cup\Delta = I$, hence,
\[dist_{\texttt{sub}}(\hat{J}\cup\emptyset,I) \leq \max\{1.5,2-0.5^{\sigma-1}\}\cdot dist_{\texttt{sub}}(J_{opt},I)\]
Without loss of generality, we have $\sigma\geq2$, then Algorithm~\ref{algo:te-lp-osr} returns a $\left(2-0.5^{\sigma-1}\right)$-approximation.
\end{proof}
\noindent\textbf{{\normalsize \textit{Remarks}}}.
Note that this ratio depends on only the size of functional dependency set other than the scale of input data.
Therefore, a simple corollary implies a ratio of 1.5 for \fdseta, \fdsetb, and \fdsetc, and 1.75 for \fdsetd, no matter how large of the input data.\par\smallskip
A naive enumeration of triad is time wasting.
In our algorithm, as in the proof of theorem~\ref{thm:2-0.5}, it is not necessary to eliminate all disjoint triads as possible.
Instead, to obtain a good ratio, it needs only eliminate all disjoint triads with respect to each single functional dependency.
Then, for each single functional dependency, sorting or hashing techniques could be utilized to speed up the triad eliminating,
and skip the finding of triads across different functional dependencies.

\subsubsection{Improved ratio by k-quasi-Tur$\acute{\text{\large \textit{a}}}$n property}
Triad elimination based \textsf{TE LP-OSR} does not capture the characteristic of input data instance.
We next give another approximation algorithm \textsf{QT LP-OSR}.
In fact, we found that constraints could be derived to strengthen LP formula.
Intuitively, for each functional dependency, each determinant equivalence class contains several determinant-dependent equivalence classes, say $k$, hence, tuples in at least $k-1$ classes should be eliminated from $I$ to obtain an s-repair.
Therefore, constraints could be invented to limit the lower bound of variables taking value $1$ according to the $k-1$ classes, so that a better ratio could be obtained for some featured cases.

Formally, consider a determinant equivalence class $[p]$ containing $m$ determinant-dependent equivalence classes $[pq_1]$, ..., $[pq_m]$, hence, \[\left|[p]\right|=\left|[pq_1]\right|+\dots+\left|[pq_m]\right|\]

\noindent\textbf{\emph{$k$-quasi-Tur$\acute{\text{a}}$n}}.
Given $k>1$, a tuple $\mathbf{t}\in{I}$ is of \emph{$k$-quasi-Tur$\acute{\text{a}}$n} property if and only if there is some functional dependency $\varphi$ and a determinant equivalence class $[p]$ with respect to $\varphi$ such that \[\mathbf{t}\in[p], {m\geq3}, \forall i, 1\le i\le m, \left|[p]\right|-\left|[pq_i]\right|\geq k\left|[pq_i]\right|\]
\begin{example}
As mentioned in Example \ref{eg:dec}, given the functional dependency $\mathsf{fd_2 : [zip] \rightarrow [CT, ST]}$, the determinant equivalence class $\mathsf{[60462]}$ is partitioned into $m = 4$ determinant-dependent equivalence classes. It is easy to verify that $\mathsf{[60462]}$ is a \emph{$2$-quasi-Tur$\acute{\text{a}}$n}
\end{example}
Then we characterize the data with parameter $\eta_{k}$ which is the portion of $k$-quasi-Tur$\acute{\text{a}}$n tuples in $I$.
A strengthened LP could be formulated as follows,
\begin{eqnarray*}
minimize & \sum_{\mathbf{t}_i\in I} x_i&\\
s.~t. &x_i\geq 0, &\forall~\mathbf{t}_i\in I\\
&  x_i+x_j\geq 1,& \forall~\{\mathbf{t}_i,\mathbf{t}_j\}\nvDash\mathsf{\Sigma},\\
&\sum\limits_{\mathbf{t}_i\in [p]} x_i> \left|[p]\right| - \max\limits_{j}{\left|[pq_{j}]\right|} - \epsilon, & \text{for every } [p].
\end{eqnarray*}

In this model, pick a small enough $\epsilon>0$, the inequality guarantees that in any integral solution, at least $|[p]|-\max_{j}{\left|[pq_{j}]\right|}$ variables taking value of $1$. However, in the fractional solution, we could not limit the number of $1$-variables, for example, a slop line cannot distinguish points $(0.5,0.5)$, $(0,1)$ and $(1,0)$.
However, even so, we will show that this number could still be limited to improve the approximation ratio.
\begin{claim}
	Every extreme point of any solution to the linear programming is in $\{0,0.5,1\}$.
\end{claim}
One can simply verify the correctness and prove it by contradiction, we omit the proof here.

Every solution of this strengthened linear programming still admits the half-integral property,
hence, we take the basic rounding strategy such that
\begin{equation*}
x_i = \left\{
\begin{array}{ll}
{0}, &\text{if}\quad{x_i=0},\\
{1}, &\text{if}\quad{x_i=0.5},\\
{1}, &\text{if}\quad{x_i=1}.
\end{array}
\right.
\end{equation*}
then, for tuples in each determinant equivalence class $[p]$, at most $\max_{j}{|[pq_{j}]|}$ variables will be rounded as $1$ wrongly.
Formally, for each determinant equivalence class $[p]$, define $S_1^{[p]}$ and $S_{0.5}^{[p]}$ as follows,
\[S_1^{[p]}\coloneqq\left\lbrace x_i|~\mathbf{t}_i\in [p], x_i=1\right\rbrace,\quad S_{0.5}^{[p]}\coloneqq\left\lbrace x_i|~\mathbf{t}_i\in [p], x_i=0.5\right\rbrace\]
then we have the following lemma,
\begin{lemma}
$\left|S_1^{[p]}\right|\geq |[p]|-2\max\limits_{j}{|[pq_{j}]|}$,\quad
$\left|S_{0.5}^{[p]}\right|\leq 2\max\limits_{j}{|[pq_{j}]|}$
\end{lemma}
\begin{proof}
	Due to the constraint \[\sum_{\mathbf{t}_i\in [p]} x_i> |[p]| - \max_{j}{|[pq_{j}]|} - \epsilon\]
	hence,
	\[\left|S_1^{[p]}\right|+0.5\left|S_{0.5}^{[p]}\right|>|[p]|-\max_{j}{|[pq_{j}]|}-\epsilon\]
	The worst case is that $|S_{0.5}^{[p]}|\leq|[p]|-|S_1^{[p]}|$, then we have,
	\[\left|S_1^{[p]}\right|+0.5(|[p]|-\left|S_{1}^{[p]}\right|)>|[p]|-\max_{j}{|[pq_{j}]|}-\epsilon,\]
	that is
	\[\left|S_1^{[p]}\right|>|[p]|-2\max_{j}{|[pq_{j}]|}-2\epsilon\]
	Pick small $\epsilon$ such that $\epsilon<0.5$, then this lemma follows.
\end{proof}

This lemma derives the a ratio depending on the portion of $k$-quasi-Tur$\acute{\text{a}}$n tuples $\eta_{k}$ where $0<\eta_{k}\leq 1$ for any $k\geq2$.

\begin{theorem}
	\textsf{\emph{QT LP-OSR}} returns a $(2-\eta_{k}+\frac{\eta_{k}}{k})$-optimal subset repair.
\end{theorem}
\begin{proof}
Let $OPT$ be the fractional optimal solution of the strengthened LP, thus
$OPT = \left|S_1\right|+0.5\left|S_{0.5}\right|$
Consider the subset $H$ of all the $k$-quasi-Tur$\acute{\text{a}}$n tuples, let the solution intersecting with $H$ is
$OPT_H=\left|S_1^{H}\right|+0.5\left|S_{0.5}^{H}\right|$.
Let $\hat{J}$ be the approximated s-repair, then in $H$, the number of tuples rounded out of the approximated s-repair,
\[(I\setminus\hat{J})\cap{H} = OPT_H + 0.5\left|S_{0.5}^{H}\right|\]
and
\[|(I\setminus{J_{opt}})\cap{H}| \geq OPT_H\]
then we derive the ratio as follows
\begin{eqnarray*}
\frac{|(I\setminus\hat{J})\cap{H}|}{|(I\setminus{J_{opt}})\cap{H}|}
&\leq&
\frac{OPT_H+0.5\left|S_{0.5}^{H}\right|}{OPT_H}\\
&\leq&
1+\min\limits_{[p]}\left\lbrace0.5\cdot \frac{2[pq_{\max}]}{|[p]|-|[pq_{\max}]|}\right\rbrace\\
&\leq&
1+\frac{1}{k}\\
\end{eqnarray*}
Then for the other part,
\[\frac{(I\setminus\hat{J})\cap I\setminus{H}}{(I\setminus{J_{opt}})\cap I\setminus{H}}\leq 2\]
Therefore we have $\frac{dist_{\texttt{sub}}(\hat{J},I)}{dist_{\texttt{sub}}(J_{opt},I)}\leq\eta_{k}(1+\frac{1}{k}) + 2(1-\eta_{k}) = 2 - \eta_{k} + \frac{\eta_{k}}{k}$
\end{proof}

Combine the approximations based on the strengthened LP and the triad elimination, a better approximation is provided.
Note that, it is polynomial-time to find a best pair $(k,\eta_k)$ to capture the data characteristic as possible, so as to improve the ratio as much as possible.

\section{Fast Estimate FD-Inconsistency Degree}
The hardness of OSR computing implies FD-inconsistency degree evaluation is also hard.
Therefore, we take effort to find an approximation of such degree.
Fortunately, an observation is that we aim to compute the ratio, but not any OSR itself, hence, to achieve a constant relative ratio, a relaxation of approximation ratio with an $O(n)$ factor is allowed.
In this section, we show a fast FD-inconsistency evaluation of subset query result.
To obtain a good approximation in sublinear complexity, we allow a relative ratio $2$ and an additional additive error $\epsilon$ where $0 < \epsilon < 1$, \ie, given an FD set \fdset and a subset query $Q$ on an instance $I$, the algorithm computes an estimation $\tilde{incDeg}(Q(I), \mathsf{\Sigma})$ such that with high constant probability such that \[incDeg(Q(I), \mathsf{\Sigma}) \le \tilde{incDeg}(Q(I), \mathsf{\Sigma}) \le 2 \cdot incDeg(Q(I), \mathsf{\Sigma}) + \epsilon.\]

\subsection{Subset Query Oracle}
As the diversity of subset queries, we model them as a $\subseteq$-oracle, such that, query complexity of the algorithm can be analyzed in terms of operations supported by the oracle.
The rest work is to find out the way of implementing the $\subseteq$-oracle for a specific subset query.
The time complexity of FD-Inconsistency evaluation for this kind of subset query then can be derived by combining query complexity and time complexity of the oracle.\par
Given an instance $I$ of a relation schema $R$ and a subset query $Q$, the corresponding $\subseteq$-oracle $O(I, Q)$ is required to answer three  queries about the result $Q(I)$:\smallskip

$\mathsf{O(I, Q)\textsf{.sample\_tuple()}}$.
Since the algorithm introduced later is sample-based, the oracle has to provide a uniform sample on the result set $Q(I)$. But sampling after the evaluation of $Q$ is incompetent to obtain a sublinear approximation, since the retrieval of $Q(I)$ will take at least linear time. A novel method of sampling is essential to implement the oracle.\smallskip

$\mathsf{{O(I, Q).in\_result(\mathbf{t})}}$.
It is to check the membership of a tuple $\mathbf{t}$ of $Q(I)$, such that, it returns true if the input tuple $\mathbf{t}$ belongs to $Q(I)$, otherwise , it returns \emph{false}.
As we shown in the next subsection, it is mostly used to check if  $Q(I)$ contains the conflict $\{\mathbf{t},\mathbf{s}\}$.\smallskip

$\mathsf{O(I, Q)\textsf{.size()}}$.
Recall the definition of $incDeg(Q(I), \mathsf{\Sigma})$, the result size is in the denominator.
It only returns the number of tuples in $Q(I)$.
Obviously, it is intolerable to compute the size by evaluating the query. \smallskip

As an example, we next show a concrete implementation of $\subseteq$-oracle for range queries.\par
\noindent\textbf{An implement of $\subseteq$-oracle.} In the following, we present an indexing-based implement of $\subseteq$-oracle for range queries. 
Without the loss of generality, let $[low, high]$ be the query range of attribute $A$, hence, the corresponding query result consists of all tuples $\mathbf{t}$ such that $low \le  \mathbf{t}.A \le high$.
Then only $B^+$-tree index is sufficient to implement the $\subseteq$-oracle.
As mentioned before, the most challenge is to implement the three operations in sublinear time.

Recall the general structure of $B^+$-tree,
each node maintains a list of key-pointer pairs.
And every node, for each key in it, say $k$, records two counters:
the number $N_{<}^k$ of tuples whose keys are less than $k$,
and the number $N_{=}^k$ of tuples whose keys are equal to $k$.
Given a range query $Q$, say $[l,h]$, on an instance $I$ with $n$ tuples,
$\subseteq$-oracle first queries the boundary (leaf) nodes $l$ and $h$
to get the size of $Q(I)$ such that $N_{=}^{h}+N_{<}^{h}-N_{<}^{l}$.
Then $\subseteq$-oracle could sample an integer $d$ uniformly from $[N_{<}^{l},N_{=}^{h}+N_{<}^{h}]$,
then fetch the tuple in $B^+$-tree by performing a binary search of $d$ where the offsets of $d$ to counters are taken as the keys.
As for verifying whether a tuple $\mathbf{t}$ is in $Q(I)$,
it is easy to make a comparison with the boundary of $Q$.
At last, the size of $Q(I)$ can be easily calculated,
since it is equal to the length of the integer interval.
Therefore, all the three operations is tractable in a logarithmic time.

In fact, another implementation is much more straightforward.
Note that tuples are arranged in specific order as a list, and the result of a range query is always a consecutive part of it.
Then, for each tuple, label it with a distinct $id$, so that each label represents its corresponding tuple.
All the $id$s are consecutive and sorted in the order induced by the selection condition.
Then it is easy to verify that all the three kinds of queries can be answered in $O(1)$ time.

Nevertheless, based on our model, one could also be free from the consideration on materialization of query result,
such as we shown for range queries, the materialization of query result could be avoided.

\subsection{Ranking and $(2,\epsilon)$-Estimation}

Recall that, a two-tuple subset $J = \{ \mathbf{t}, \mathbf{s} \}$ is a conflict if and only if $\mathbf{s} \in [\mathbf{t}.\mathsf{X}] \setminus [\mathbf{t}.\mathsf{XY}]$ with respect to some FD: $\mathsf{X} \rightarrow \mathsf{Y}$ in $\mathsf{\Sigma}$.
Let $C$ be the set of all conflicts in an instance $I$, then an \textit{s-repair} $S$ of $I$ can be derived in the following way:
ranking all the conflicts of $C$ in an ascendant order $\Pi$,
\emph{pick} the current first conflict $J=\{\mathbf{t},\mathbf{s}\}$ and \emph{remove} tuples $\mathbf{t}$ and $\mathbf{s}$ from $I$,
then \emph{eliminate} all the conflicts containing $\mathbf{t}$ or $\mathbf{s}$ from $C$,
repeat the \emph{pick-remove-eliminate} procedure until no any conflict left in $C$,
then the $I$ left is a repair $S$.

We claim that $S$ is a 2-optimal s-repair of $I$.
The proof is quite straightforward, observe that,
any repair has to eliminate at least one tuple of the conflicts picked in such procedure,
then we have
\begin{equation*}
\frac{1}{2} dist_{\texttt{sub}}(S, I) = \frac{1}{2}(|I| - |S|) \le dist_{\texttt{sub}}(S_{opt}, I),
\end{equation*}
thus achieving a 2-optimal s-repair.


By applying Chernoff bound, if we uniformly sample $p = \Theta(\frac{1}{\epsilon^2})$ tuples, and count the number of tuple not in $S$, say $q$, then with high probability \[|S| - \frac{\epsilon}{2} n \le \frac{p - q}{p} \cdot n \le |S| + \frac{\epsilon}{2} n.\] Hence we can obtain a $(2, \epsilon)$-approximation of $incDeg(I,\mathsf{\Sigma})$ as defined previously. And observe each ranking $\Pi$ decides a 2-optimal s-repair $S$ so that we could scan the ranking once, \emph{verify} the membership of each tuple in $S$, and count the number $q$, however, in such a trivial way, it takes a linear time complexity. In the next subsection, we will give a sublinear time implementation of the verification step.

We argue that the sampling method mentioned above still works for the 2-optimal s-repair of $Q(I)$ with the same probabilistic error bound.
\begin{figure}[h]
\begin{small}
\centering
\vspace{0mm}
\includegraphics[height=40mm]{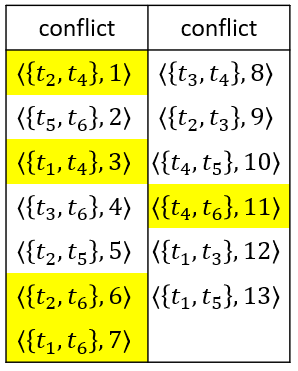}
\vspace{-3mm}
\caption{Ranking of conflicts in $I$ and Ranking of conflicts in $Q(I)$ (yellow shaded)}
\label{fig:eg_qicfs}
\vspace{0mm}
\end{small}
\end{figure}
Consider a subset query $Q$ on $I$,
a conflict $J$ in any $Q(I)$ is still a conflict in $I$, and the ranking induced by any $Q(I)$ from $\Pi$ is still ascendant.
Continue with Example \ref{eg:incDeg}, given a range query $Q = [15, 45]$ on \textsf{PR}, the result set $Q(I)$ is $\{\mathbf{t_1},\mathbf{t_2},\mathbf{t_4},\mathbf{t_6}\}$. As shown in the figure \ref{fig:eg_qicfs}, all conflicts in $Q(I)$ are yellowed faded and the ranking induced by them from $\Pi$ is still ascendant.
Let $S$ be an optimal s-repair of $Q(I)$, then ranking $\Pi$ could be reused to compute a $(2, \epsilon)$-approximation of $incDeg(Q(I),\mathsf{\Sigma})$.

\subsection{Fast Estimate FD-inconsistency Degree}
We first settle the preprocessing method, then show an efficient implementation of the \emph{sample-and-verify} method for subset queries.
\subsubsection{Preprocessing}
Ranking-based method mentioned above implies a pre-defined rank could be reused for the FD-inconsistency degree evaluation of every subset query result.
Therefore, we discover all the conflicts in $I$ with respect to $\mathsf{\Sigma}$ in the preprocessing step, and assign a unified rank to $C$ in advance.
That is a distinct rank $r$ is assigned to each conflict $J=\{ \mathbf{t}, \mathbf{s}\}$ as shown in figure \ref{fig:eg_cfs}.\par
\begin{algorithm}  
	\caption{\textsf{Preprocessing Procedure}}  
	\label{alg:D}  
	\hspace*{0.02in}{\bf Input:}
	An instance $I$ of a relation schema $R$, and a set of functional dependencies $\mathsf{\Sigma}$\\
	\hspace*{0.02in}{\bf Output:} 
	Set $C$ of conflicts in $I$ and a ascendant ranking $\Pi$ on $C$
	
	\begin{algorithmic}[1] 
		\FOR {each two tuples subset $J = \{ \mathbf{t}, \mathbf{s} \}$ of $I$}
		\IF {$J$ is a $\varphi$-\emph{conflict} for some $\varphi \in \mathsf{\Sigma}$ in $I$}
		\STATE {Generate a unique and no duplicated ranking $r$;}
		\STATE {Append $\langle \{ \mathbf{t}, \mathbf{s} \}, r \rangle$ to $C$;}
		\ENDIF
		\ENDFOR
		\STATE {Sort $C$ according to $r$;}
	\end{algorithmic}  
\end{algorithm}
Algorithm \ref{alg:D} illustrates the preprocessing procedure.
Let $n$ be the size of $I$.
The running time of Preprocessing is at worst $O(n^2\log{n})$.
Because there are at most $O(n^2)$ tuple pairs of $I$ and we consider data complexity in this paper,
it takes $O(n^2)$ time to find out all conflicts of $I$.
And Step 7 may take $O(n^2\log{n})$ time.
Note that, the number of conflicts are usually not that large in practice, techniques like hash-based partition could be taken as a tool to find all possible conflicts,
so that the time cost of preprocessing could be further lower,
but we do not emphasis them in this paper.

\subsubsection{Verification Locally}
Recall the \emph{sampling-and-verification} procedure,
for any tuple $\mathbf{t}\in{Q(I)}$ sampled uniformly,
it is to check if $\mathbf{t}$ is in the 2-optimal s-repair of $Q(I)$ derived by given ranking $\Pi$.
During this procedure, every conflict $J$ in $Q(I)$ eliminated in the checking procedure has a lowest rank when we turn to check it.
That is, any conflict $J^\prime$ in $Q(I)$ intersecting with $J$ are either already eliminated or having a rank higher than $J$.
It is easy for the sequential implementation if we scan the ranking from its beginning.
However, it is difficult without scanning the entire ranking, since the sampled tuple may not locate in the beginning.

To enable a sublinear evaluation, we need a \emph{start-from-anywhere} implement method.
Fortunately, the locality of a ranking can be utilized to avoid scanning the entire ranking.
We employ a recursive verification starting from the conflicts involving current sampled tuple.

Basically, we begin with the sampled tuple $\mathbf{t}$ and check the conflicts in $Q(I)$ caused by $\mathbf{t}$ in turn from rank lowest to highest. For each conflict $J$ we currently considering, $J$ should be eliminated if one of the following conditions holds,
\begin{enumerate}
\item[-] $J$ is lowest among all the conflicts in $Q(I)$ intersecting with it,
\item[-] Every conflict $J^\prime$ in $Q(I)$ intersecting with $J$ and lower than $J$ are already known to be eliminated.
\end{enumerate}
Otherwise, recursively check $J^\prime$ by the same procedure.
At last, if none of conflicts in $Q(I)$ containing $\mathbf{t}$ is decided to be eliminated, then $\mathbf{t}$ should stay in $S$, otherwise not, and count it into $q$.
The correctness of this method is obviously,
the only concern is the running time which depends on the number of recursive calls.
We next formally describe it and bound the number of recursive calls.
\subsubsection{Sublinear Estimation}
Algorithm \ref{alg:E} \textsf{Fast-IncDeg} performs $s$ sampling-and-verifying operations and counts the number of tuples sampled but not in $S$ by calling \textbf{function} NotInSR$(\mathbf{t}, C)$.
Since $S$ is a 2-optimal s-repair of $I$, \textsf{Fast-IncDeg} outcomes an $(2,\epsilon)$-estimation of FD-inconsistency degree of $Q(I)$.
\begin{algorithm}  
	\caption{\textsf{Fast-IncDeg}}
	\label{alg:E}  
	\hspace*{0.02in}{\bf Input:}
	Set $C$ of conflicts in $I$, subset query $Q$, error $\epsilon$\\
	\hspace*{0.02in}{\bf Output:} 
	$\tilde{incDeg}(Q(I),\mathsf{\mathsf{\Sigma}})$
	
	\begin{algorithmic}[1] 
		\STATE {$\tilde{dist}_{\texttt{sub}} := 0$;}
		\FOR {$i := 1$ to $8 / \epsilon^2$}
		\STATE {$\mathbf{t} := O(I, Q)\mathsf{.sample\_tuple()}$;}
		\IF {NotInSR$(\mathbf{t}, C)$}
		\STATE {$\tilde{dist}_{\texttt{sub}} := \tilde{dist}_{\texttt{sub}} + 1$;}
		\ENDIF
		\ENDFOR
		\RETURN {$\frac{\tilde{dist}_{\texttt{sub}}}{O(I, Q)\mathsf{.size()}} + \frac{\epsilon}{2}$;}
	\end{algorithmic}  
\end{algorithm}
\floatname{algorithm}{Subroutine}
\setcounter{algorithm}{0}
\begin{algorithm}  
	\caption{NotInSR$(\mathbf{t}, C)$}
	\label{sub:A}
	\begin{algorithmic}[1] 
		\STATE {Let $\langle \{ \mathbf{t}, \mathbf{t_1}\}, r_1 \rangle, \cdots, \langle \{ \mathbf{t}, \mathbf{t}_l \}, r_l \rangle$ be the tuples in $C$ including $\mathbf{t}$ in order of increasing $r$;}
		\FOR {$i := 1$ to $l$}
		\IF {$O(I, Q)\mathsf{.in\_result}(\mathbf{t}_i)$ \AND \textsf{Eliminate}$(\{ \mathbf{t}, \mathbf{t}_i \})$}
		\RETURN {\TRUE;}
		\ENDIF
		\ENDFOR
		\RETURN {\FALSE;} \smallskip
		
		\STATE {\textbf{function} \textsf{Eliminate}$(\{\mathbf{t}, \mathbf{s}\})$}
		\STATE {Let $\langle \{ \mathbf{t}_1, \mathbf{s}_1 \}, r_1, \rangle, \cdots, \langle \{ \mathbf{t}_l, \mathbf{s}_l \}, r_l \rangle$ be the tuples in $C$ such that $\mathbf{t}_i \in \{\mathbf{t}, \mathbf{s}\}$ in order of increasing $r$;}
		\WHILE {$r_i < r$}
		\IF {$O(I, Q)\mathsf{.in\_result}(\mathbf{s}_i)$ \AND \textsf{Eliminate}$(\{ \mathbf{t}_i, \mathbf{s}_i \})$}
		\RETURN {\FALSE;}
		\ENDIF
		\STATE {$i := i + 1$;}
		\ENDWHILE
		\RETURN {\TRUE;}
		\STATE {\textbf{end function}}
	\end{algorithmic}  
\end{algorithm}

Subroutine~\ref{sub:A} implements the \emph{verification} by calling \textsf{Eliminate} recursively.
Namely, as introduced in subsection-\emph{4.3.2}, give a conflict $J = \{\mathbf{t}, \mathbf{s}\}$, it considers all conflicts which include $\mathbf{t}$ or $\mathbf{s}$ with lower ranking.
If there are no such conflicts, it returns \emph{true}.
Otherwise, it performs recursive calls to these conflicts in the order of their ranking.
If any one-step recursion returns \emph{true}, it returns \emph{false}; Otherwise, it returns \emph{true}.
With the help of \textsf{Eliminate}, the subroutine \ref{sub:A} checks if a tuple $\mathbf{t}$ belongs to $S$.
Concretely, it performs \textsf{Eliminate} on all conflicts in $Q(I)$ including $\mathbf{t}$ in the order of their rankings,
and if there exists a conflict such that \textsf{Eliminate} returns \emph{true}, it returns \emph{true}; otherwise, it returns \emph{false}.

Now, we bound the number of recursive calls of \textsf{Eliminate}.
First, we derive an important corollary from \cite{onak2012near}.
For a ranking injection $\pi : J \rightarrow r$ of all conflicts of an instance $I$ and a tuple $\mathbf{t} \in I$, let $N(\pi, \mathbf{t})$ denote the number of conflicts that a call $\mathrm{Eliminate}(J)$ was made on in the course of the computation of $\mathrm{NotInSR}(\mathbf{t})$.
Let $\Pi$ denote the set of all ranking injections $\pi$ over the conflicts of $I$.
Given a tuple $\mathbf{t}$, let  $\delta_\mathbf{t}$ be the number of conflicts containing $\mathbf{t}$.
Then we define the maximum conflict number of $I$ as $\delta_I = \max\limits_{\mathbf{t}\in I} \{\delta_\mathbf{t}\}$
The average value of $N(\pi, \mathbf{t})$ taken over all ranking injections $\pi$ and tuples $\mathbf{t}$ is $O(\delta_I^2)$, \ie,
\begin{equation}
\label{eq:onk}
\frac{1}{m!}\cdot\frac{1}{n!}\cdot\sum\limits_{\pi \in \Pi}\sum\limits_{\mathbf{t} \in I} N(\pi, \mathbf{t}) = O(\delta_I^2)
\end{equation}

\begin{theorem}
	\label{thm:alg}
	Algorithm \textsf{Fast-IncDeg} returns an estimate $\tilde{incDeg}(Q(I), \mathsf{\Sigma})$ with a probability at least $2/3$ such that,
	\begin{equation*}
	incDeg(Q(I), \mathsf{\Sigma}) \le \tilde{incDeg}(Q(I), \mathsf{\Sigma}) \le 2 \cdot incDeg(Q(I), \mathsf{\Sigma}) + \epsilon.
	\end{equation*}
	The average query complexity taken over all rankings $\pi$, subset queries $Q$ and tuples $\mathbf{t}$ of $Q(I)$ is $O(\frac{\delta_{I}^2}{\epsilon^2})$, where the algorithm uses only queries supported by the $\subseteq$-oracle.
\end{theorem}

\begin{proof}
	By applying an additive Chernoff bound, suppose that it is sampled uniformly and independently $s = \Theta(\frac{1}{\epsilon^2})$ tuples $\mathbf{t}$ from $Q(I)$, with probability more than $2/3$,
	\begin{equation}
	\frac{dist_{\texttt{sub}}(S, Q(I))}{|Q(I)|} - \frac{\epsilon}{2} \le \frac{\tilde{dist}_{\texttt{sub}}}{|Q(I)|} \le \frac{dist_{\texttt{sub}}(S, Q(I))}{|Q(I)|} + \frac{\epsilon}{2}.
	\end{equation}
	And with the fact that $S$ is a 2-optimal \textit{s-repair}, it is obtained that,
	\begin{equation}
	incDeg(Q(I), \mathsf{\Sigma}) \le \tilde{incDeg}(Q(I), \mathsf{\Sigma}) \le 2 \cdot incDeg(Q(I), \mathsf{\Sigma}) + \epsilon.
	\end{equation}
	
	For query complexity, we first bound the number of calls of $\mathsf{Eliminate()}$.
	Given the result $Q(I)$ of a subset query $Q$, let $n^\prime$ be the number of tuples in $Q(I)$, and $m^\prime$ be the number of conflicts contained in $Q(I)$, and  the maximum conflict number of $Q(I)$.
	Now, consider the ranking $\Pi^\prime$ induced by $Q(I)$ from $\Pi$, then equation \ref{eq:onk} implies,
	\begin{equation*}
	\frac{1}{m'!}\cdot\frac{1}{n'}\cdot\sum\limits_{\pi\in\Pi'}\sum\limits_{\mathbf{t} \in Q(I)} N(\pi, \mathbf{t}) = O(\delta_{Q(I)}^2)
	\end{equation*}
	Notice that since the conflicts in $Q(I)$ is a subset of the conflicts in $I$, for each $\pi' \in \Pi'$, there are $\frac{m!}{m'!}$ number of $\pi \in \Pi$ can produce the same ranking on the conflicts of $Q(I)$. Group $\Pi$ into $m'!$ groups $\{\Pi_1, \cdots, \Pi_{m'!}\}$, and for each $\pi\in\Pi_i$ and a fixed $\mathbf{t} \in Q(I)$, $N(\pi, \mathbf{t})$ has the same value.
	So we have,
	\begin{equation*}
	\begin{split}
	&\frac{1}{m!}\sum\limits_{\pi\in\Pi}\frac{1}{|Q(I)|}\sum\limits_{\mathbf{t}\in Q(I)}N(\pi, \mathbf{t})
	\\= & \frac{1}{m!}\sum\limits_{\pi \in \{\Pi_1, \cdots, \Pi_{m'!}\}}\frac{1}{|Q(I)|}\sum\limits_{\mathbf{t}\in Q(I)}N(\pi, \mathbf{t})
	\\= & \frac{1}{m!}\cdot\frac{m!}{m'!}\sum\limits_{\pi\in\Pi'}\frac{1}{Q(I)}\sum\limits_{\mathbf{t}\in Q(I)}N(\pi, \mathbf{t})
	\\= & O(\delta_{Q(I)}^2)
	\end{split}
	\end{equation*}
	
	The we could derive the query complexity.
	Let $\mathcal{Q}$ be the space of queries, then for any $Q(I)$, we have $\delta_{Q(I)} \le \delta_{I}$, so that the \emph{average query complexity} is that
	\begin{equation*}
	\begin{split}
	&\frac{1}{m!}\sum\limits_{\pi\in\Pi}\frac{1}{|\mathcal{Q}|}\sum\limits_{Q \in \mathcal{Q}}\frac{1}{|Q(I)|}\sum\limits_{\mathbf{t}\in Q(I)}N(\pi, \mathbf{t})
	\\= &\frac{1}{|\mathcal{Q}|}\sum\limits_{Q\in \mathcal{Q}}\frac{1}{m!}\sum\limits_{\pi\in\Pi}\frac{1}{|Q(I)|}\sum\limits_{\mathbf{t}\in Q(I)}N(\pi, \mathbf{t})
	\\= &\frac{1}{|\mathcal{Q}|}\sum\limits_{Q\in \mathcal{Q}}O(\delta_{Q(I)}^2)
	\\ \le &O(\delta_I^2)
	\end{split}
	\end{equation*}
	
	Obviously, there are $O(\frac{1}{\epsilon^2})$ calls to sample a tuple from the result set. For each sampled tuple $\mathbf{t}$, the average number of calls to $\mathsf{Eliminate}$ is $O(\delta_{I}^2)$. So the number of calls to $\mathsf{in\_result()}$ is $O(\frac{\delta_{I}^2}{\epsilon^2})$.
\end{proof}

In addition, inspired by the methodology proposed in \cite{onak2012near}, a pre-defined ranking in the preprocessing can be saved by ranking on the fly, that is, we only need to discover all conflicts in the preprocessing step and ranking whenever it is required.
Since the basic idea is similar with our method, we omit the detail here, instead, we compare the two different implements in our experiments to show the efficiency of our method.

\eat{
\section{Local Inconsistency Computing}
\subsection{Subset Query Oracle}
As the diversity of subset queries, we model them as a $\subseteq$-oracle, such that,
\subsection{Sublinear Estimation}
}
\section{Experiments}

This section experimentally evaluates the performace of our algorithms for OSR computing and FD-inconsistency evaluation.

\subsection{Experimental Settings}

All experiments are conducted on a machine with eight 16-core Intel Xeon processors and 3072GB of memory.\smallskip

\noindent\textbf{Dataset.}
We used two datasets to evaluate the performance of algorithms for OSR computing and FD-Inconsistency evaluaction experimentally. \smallskip

\textit{Dataset 1}: \textsc{order} \emph{data} is an instance of the schema \textsf{order} shown in Example \ref{eg:sch&fd}. Our set $\mathsf{\Sigma}$ consists of 4 FDs taken from Example \ref{eg:sch&fd}. To populate the relation we scraped product informations from \textsc{amazon} and collected real-life data: the zip and area codes for major cities and twons for all US states\footnote{http://www.geonames.org/} and street informations for all the United States\footnote{http://results.openaddresses.io/}. We generated datasets of various size, ranging from 10M to 100M tuples.\smallskip

\textit{Dataset 2}: \textsc{dblp} \emph{data} was extracted from \textsc{dblp} Bibliography\footnote{https://dblp.org/xml/}. It consists of 40M tuples and the format is as follows:
\vspace{-4mm}
\begin{equation*}
\mathsf{dblp(title, authors, year, publication, pages, ee, url)}
\vspace{-1mm}
\end{equation*}

Each \textsc{dblp} tuple contains the title of an article, the authors and the information of publication (year, publication venue, pages, electronic edition and, url)
We designed 4 FDs for \textsc{dblp}.
\begin{equation*}
\begin{aligned}
&\mathsf{fd_1 : [title] \to [author]} \quad &\mathsf{fd_2} &\mathsf{: [ee] \to [title]}\\
&\mathsf{fd_3 : [year, publication, pages] \to [title, ee, url]} &\mathsf{fd_4} &\mathsf{: [url] \to [title]} \\
\end{aligned}
\vspace{-1mm}
\end{equation*}

To add noise to a dataset, we randomly selected an attribute of a "correct" tuple and changed it either to a close value or to an existing value taken from another tuple. We appended such "dirty" tuples which violate at least one or more functional dependencies to the dataset. We set a parameter $\rho$ ranging from $1\%$ to $10\%$ to control the noise rate.\smallskip

\noindent\textbf{Methods.}
We implemented the following algorithms: (a) the basic approximation algorithm \textsf{BL LP-OSR} and the improved approximation algorithm by triad elimination \textsf{TE LP-OSR} for OSR computing; (b) the sublinear estimation algorithm \textsf{Fast-IncDeg} based on two implements of $\subseteq$-oracle for range query with $O(1)$ and $O(\log{n})$ time complexity respectively, and its variation mentioned in the subsection-\emph{4.3.3} \textsf{Fast-IncDeg\_ol} for FD-Inconsistency evaluation. Hence, there are totally 4 implements of Algorithm \ref{alg:E} for range query denoted by $\mathsf{\textsf{Fast-IncDeg}_{c}}$, $\mathsf{\textsf{Fast-IncDeg}_{log}}$, $\mathsf{\textsf{Fast-IncDeg}_{c}\_ol}$ and $\mathsf{\textsf{Fast-IncDeg}_{log}\_ol}$ respectively.\smallskip

\noindent\textbf{Metrics.} Since a dataset $I$ with $n$ tuples is polluted by appending $\rho n$ dirty tuples, where $\rho$ is noise rate,  the number of tuples in the optimal repair $S_{opt}$ must be larger than $n$, \ie, $dist_{\texttt{sub}}(S_{opt}, I) \le \rho n$. Hence, we calculate $dist_{\texttt{sub}}(\hat{J}, I)$ of \textsf{BL LP-OSR} and \textsf{TE LP-OSR} and use $2\rho n$ to evaluate the approximation ratio of them. What's more, according to the definition of FD-inconsistency degree, we treat $2 \rho + \epsilon$ as the upper bound of FD-inconsistency degree to ensure the correctness of the Algorithm \ref{alg:E}. To evaluate the efficiency of Algorithm \ref{alg:E}, we issue 300 queries for each algorithm and each parameter set, and record the average of the query time.

\subsection{Experimental Results}

We report our findings concerning about the accuracy and efficiency of our algorithms.

\begin{figure*}[!t]
	\begin{small}
		\centering
		\vspace{0mm}
		\subfigure[\textsc{order} $\rho  = 0.03$]{
			\hspace{-3mm}\includegraphics[height=35mm]{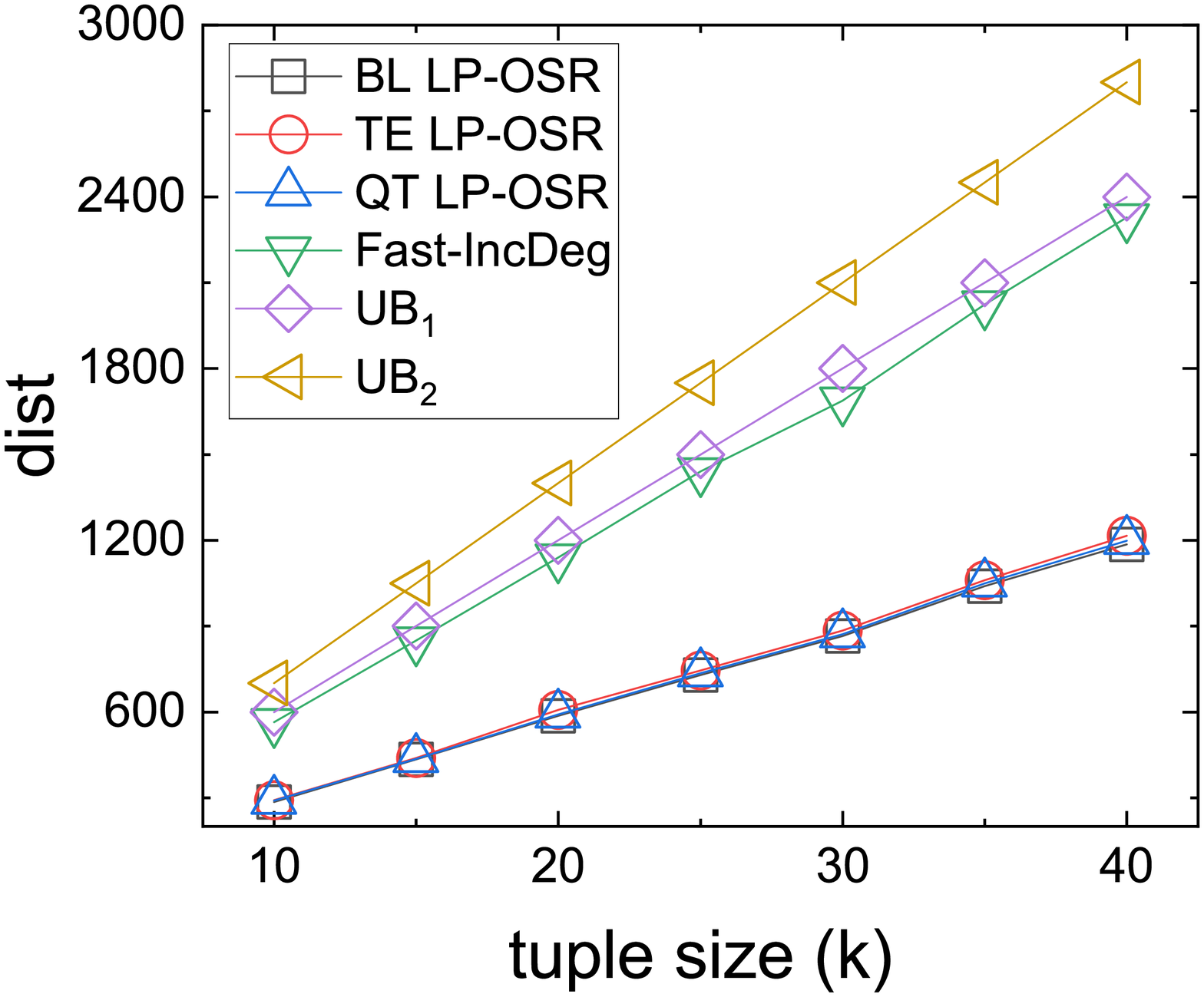} \hspace{-3mm}
			\label{fig:lp-ord-rho1}
		}
		\subfigure[\textsc{order} $\rho = 0.08$]{
			\hspace{-3mm}	\includegraphics[height=35mm]{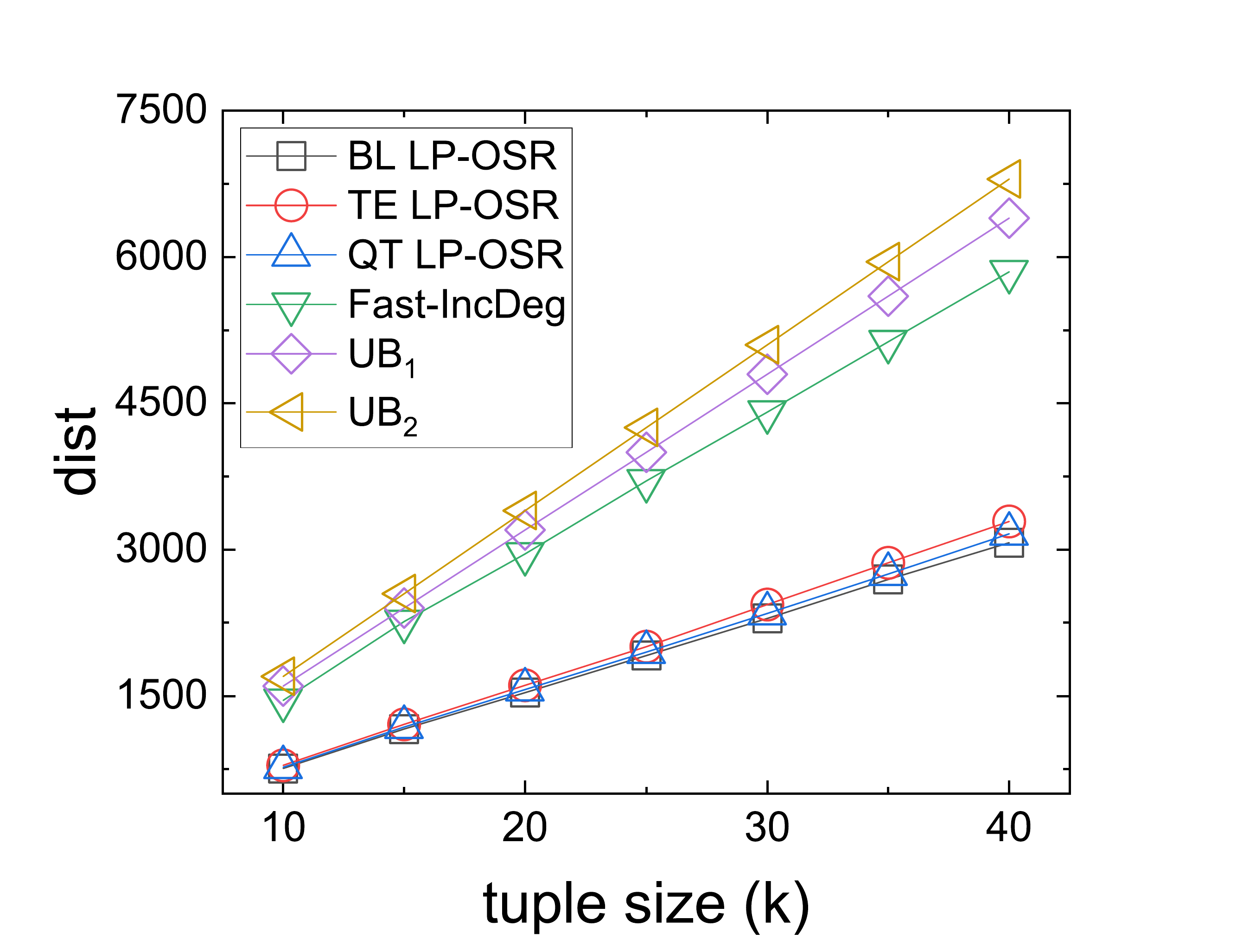} \hspace{-3mm}
			\label{fig:lp-ord-rho2}
		}
		\subfigure[\textsc{order} $n = 10$K]{
			\hspace{-3mm}	\includegraphics[height=35mm]{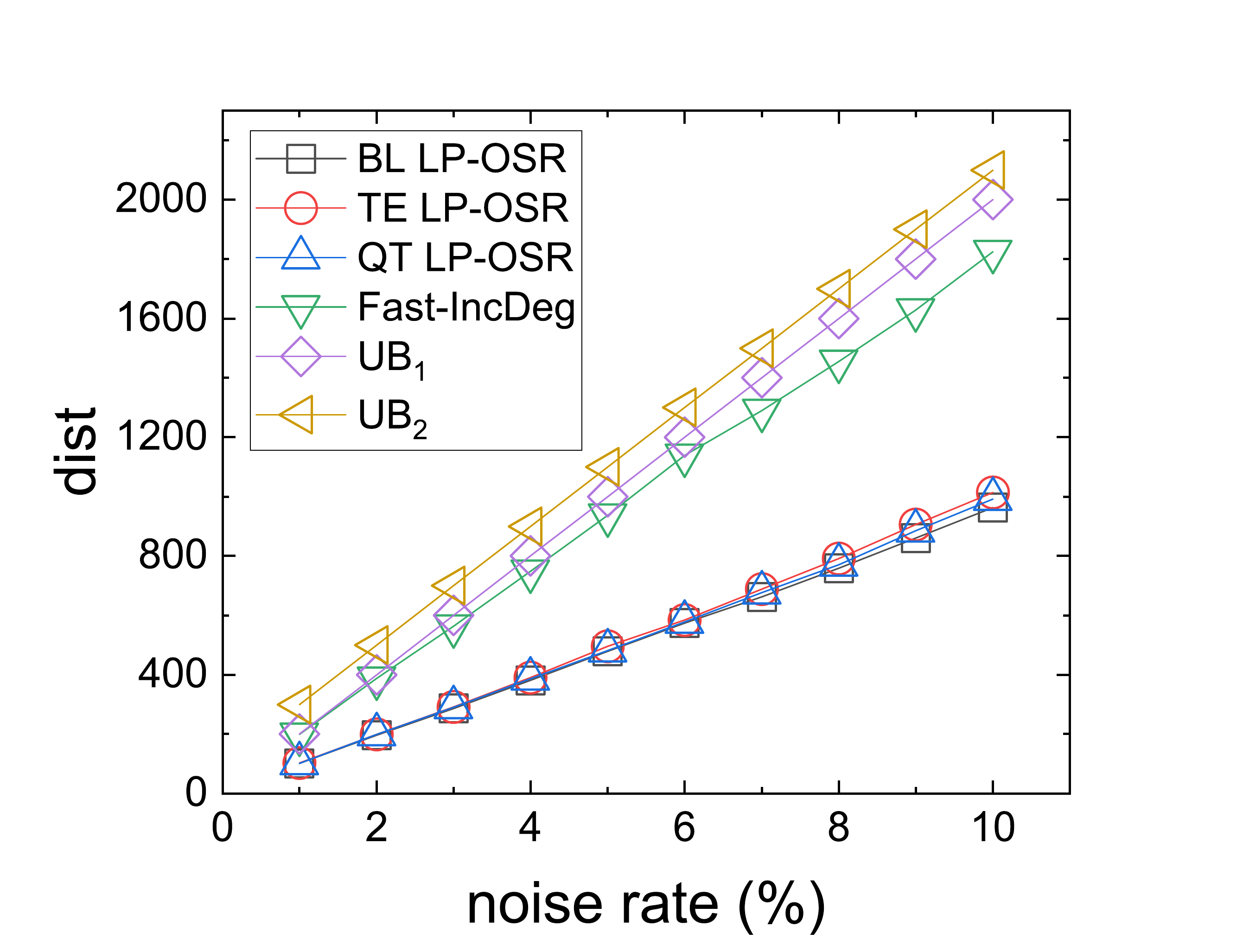} \hspace{-3mm}
			\label{fig:lp-ord-n1}
		}
		\subfigure[\textsc{order} $n = 40$K]{
			\hspace{-3mm}	\includegraphics[height=35mm]{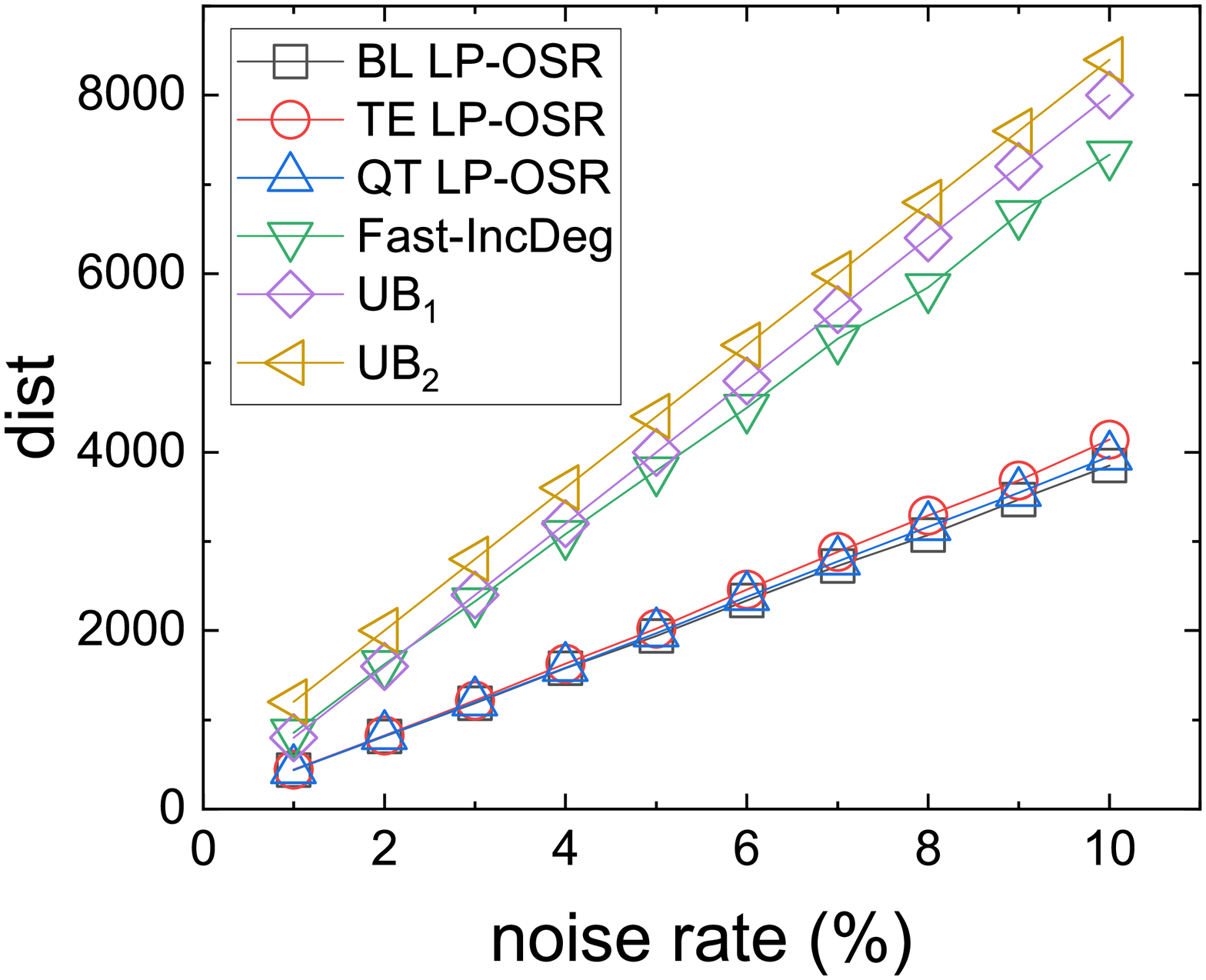} \hspace{-3mm}
			\label{fig:lp-ord-n2}
		}
		\vspace{-4mm}
		
		\subfigure[\textsc{dblp} $\rho  = 0.03$]{
			\hspace{-3mm}\includegraphics[height=35mm]{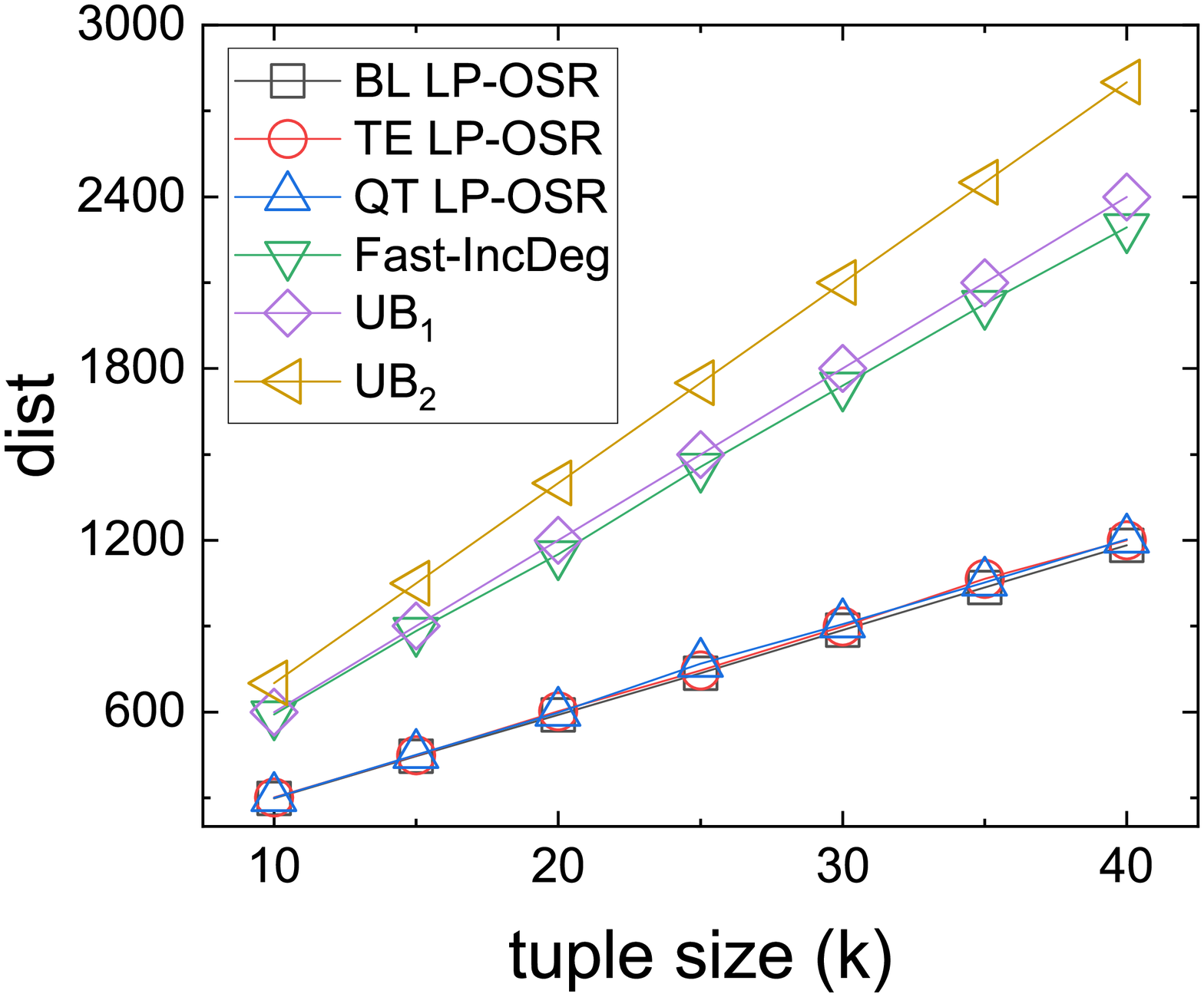} \hspace{-3mm}
			\label{fig:lp-dblp-rho1}
		}
		\subfigure[\textsc{dblp} $\rho = 0.08$]{
			\hspace{-3mm}	\includegraphics[height=35mm]{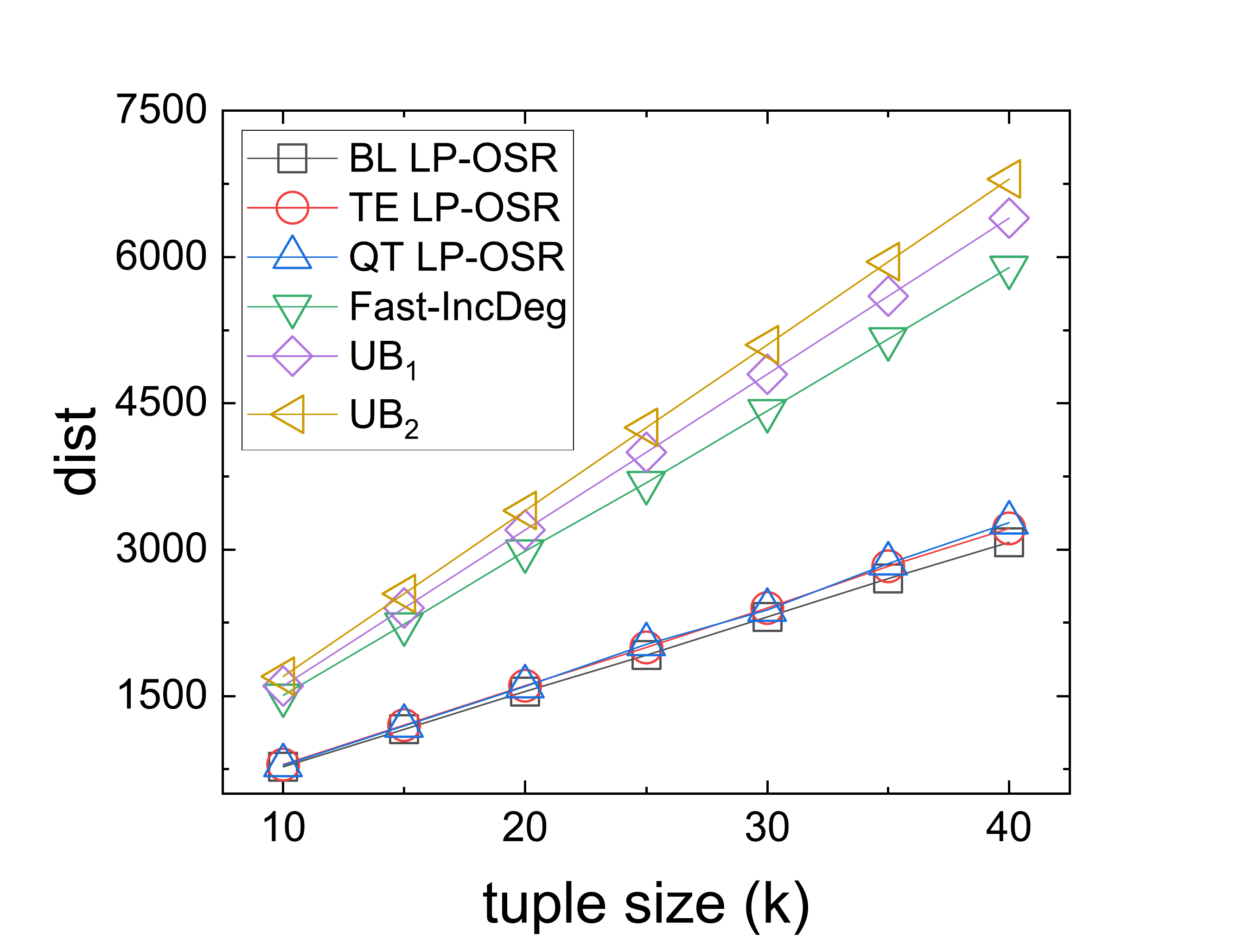} \hspace{-3mm}
			\label{fig:lp-dblp-rho2}
		}
		\subfigure[\textsc{dblp} $n = 10$K]{
			\hspace{-3mm}	\includegraphics[height=35mm]{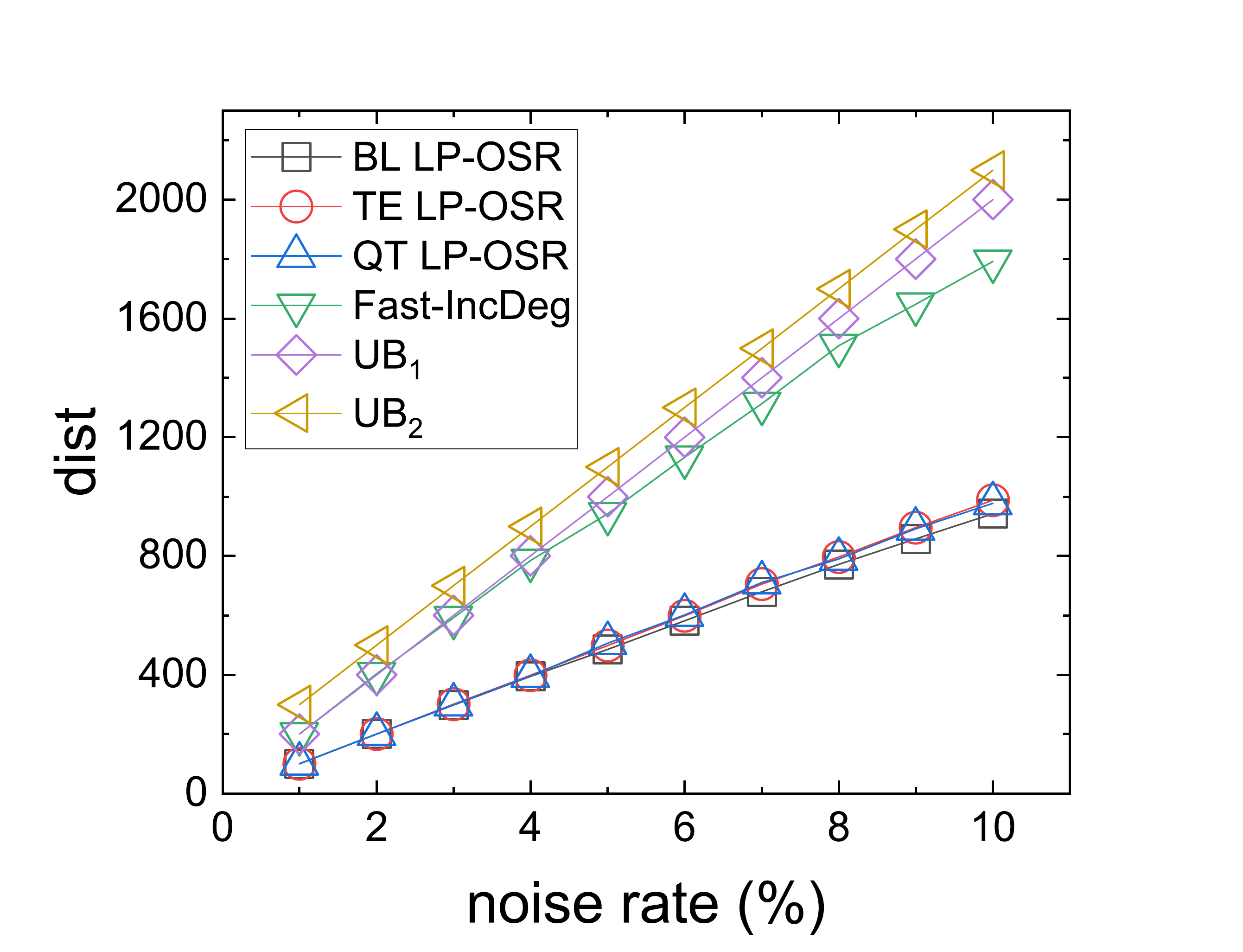} \hspace{-3mm}
			\label{fig:lp-dblp-n1}
		}
		\subfigure[\textsc{dblp} $n = 40$K]{
			\hspace{-3mm}	\includegraphics[height=35mm]{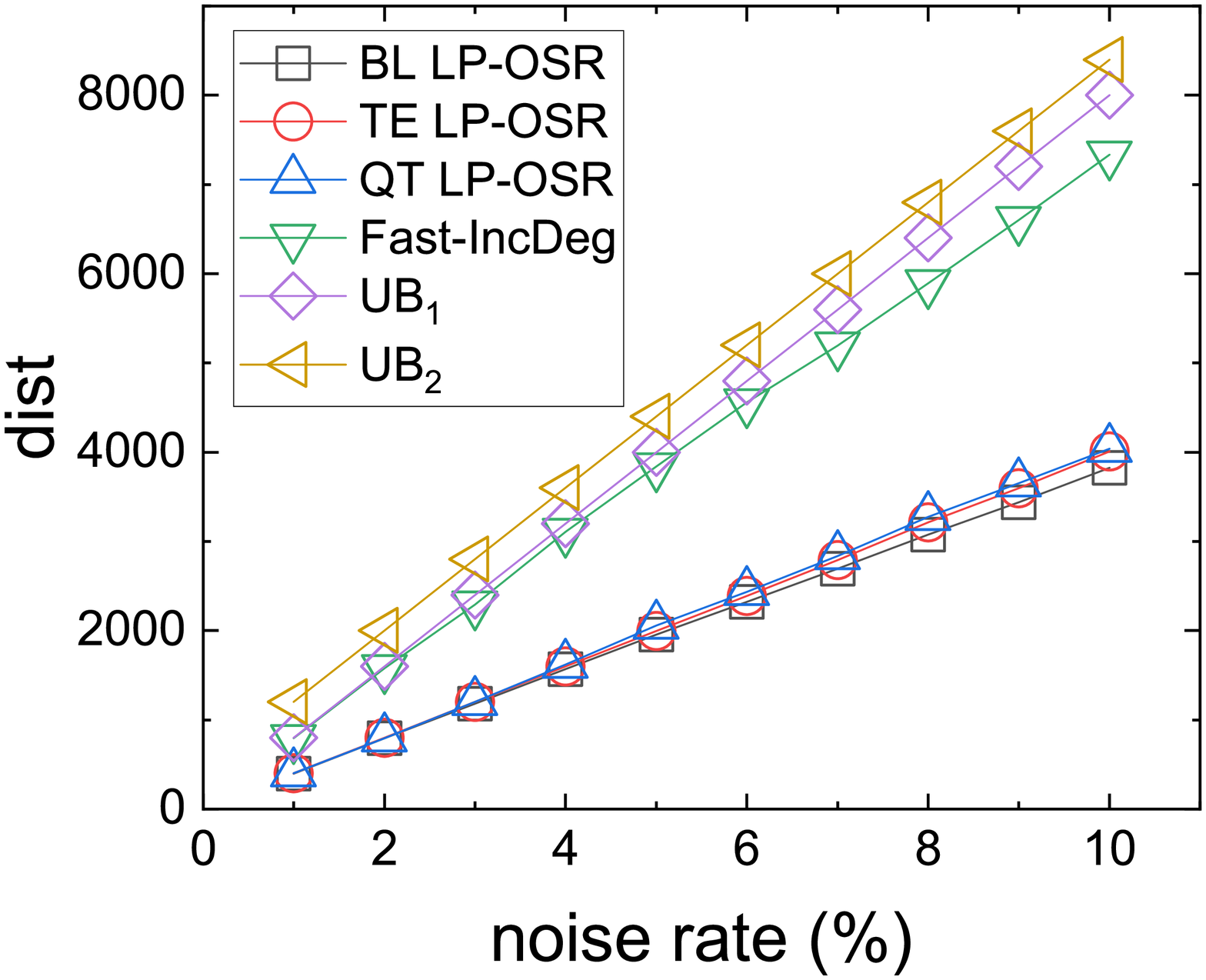} \hspace{-3mm}
			\label{fig:lp-dblp-n2}
		}
		\vspace{-3mm}
		\caption{$dist$ with different $\rho$, $n$ and $\sigma$}
		\label{fig:lp-exp}
		\vspace{0mm}
	\end{small}
\end{figure*}

\begin{figure*}[!t]
	\begin{small}
		\centering
		\vspace{0mm}
		\subfigure[\textsc{order} $\rho  = 0.05$]{
			\hspace{-3mm}\includegraphics[height=35mm]{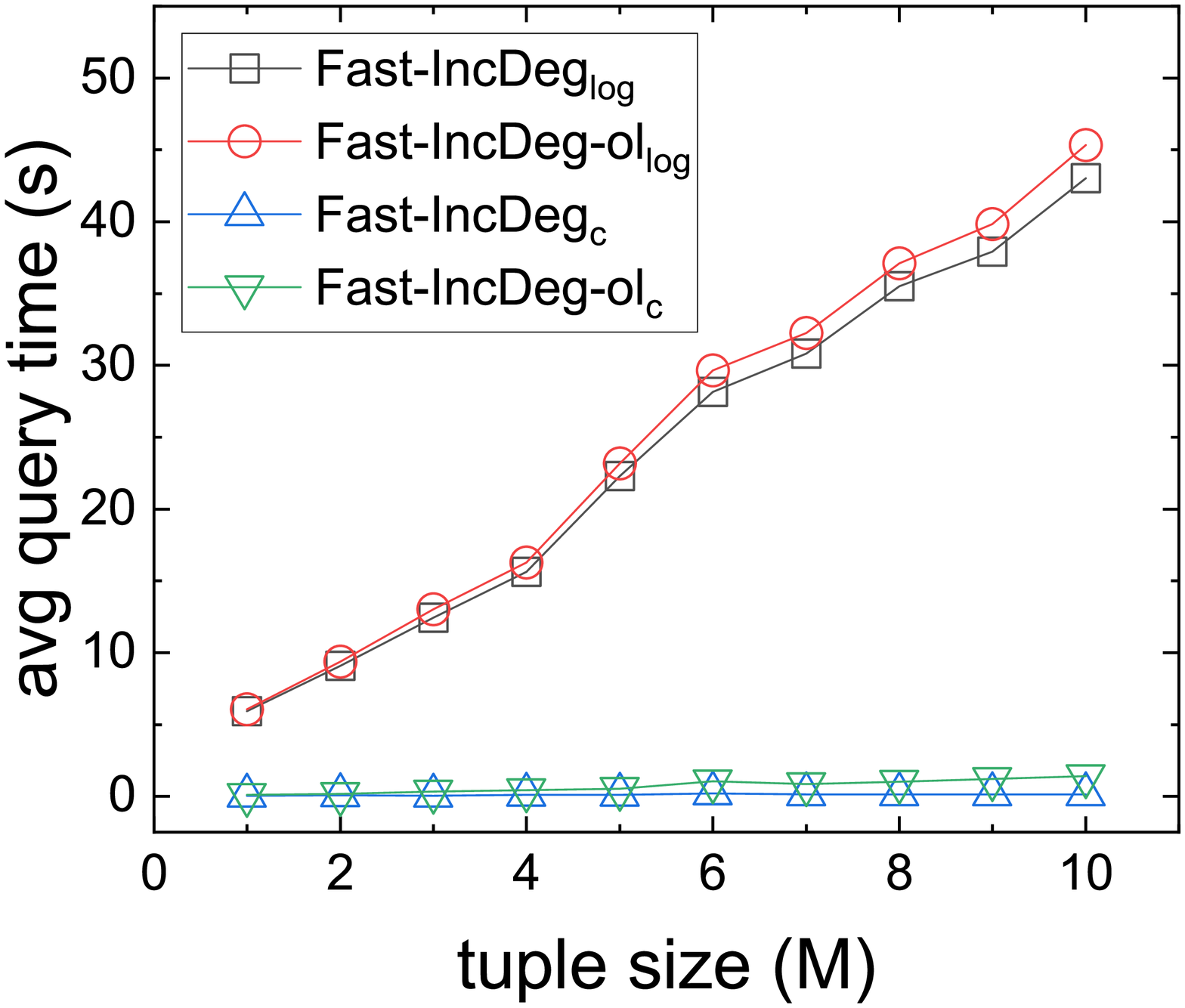} \hspace{-3mm}
			\label{fig:fi-ord-rho}
		}
		\subfigure[\textsc{order} $n = 10$M]{
			\hspace{-3mm}	\includegraphics[height=35mm]{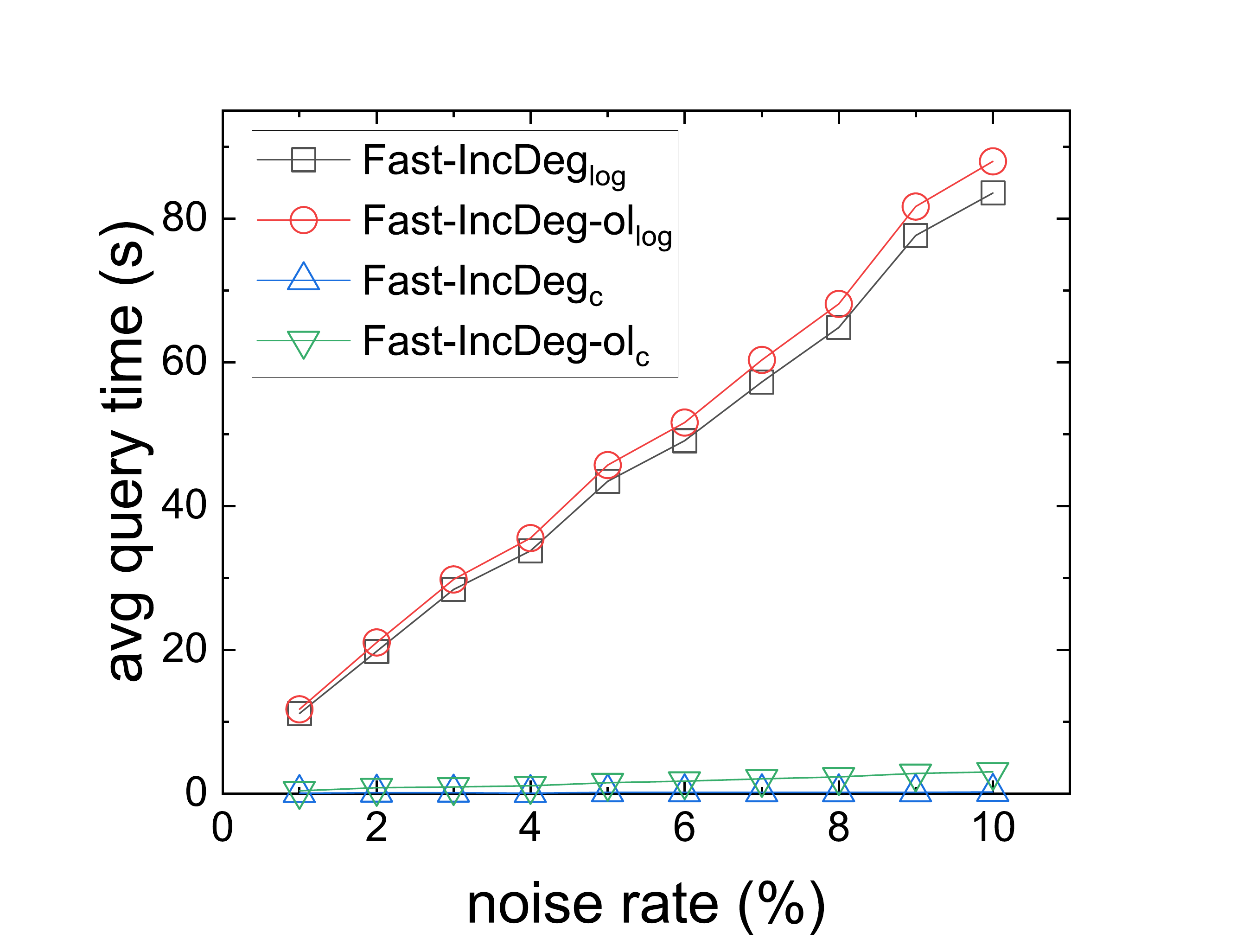} \hspace{-3mm}
			\label{fig:fi-ord-n}
		}
		\subfigure[\textsc{order} $d = 50, \rho = 0.03$]{
			\hspace{-3mm}	\includegraphics[height=35mm]{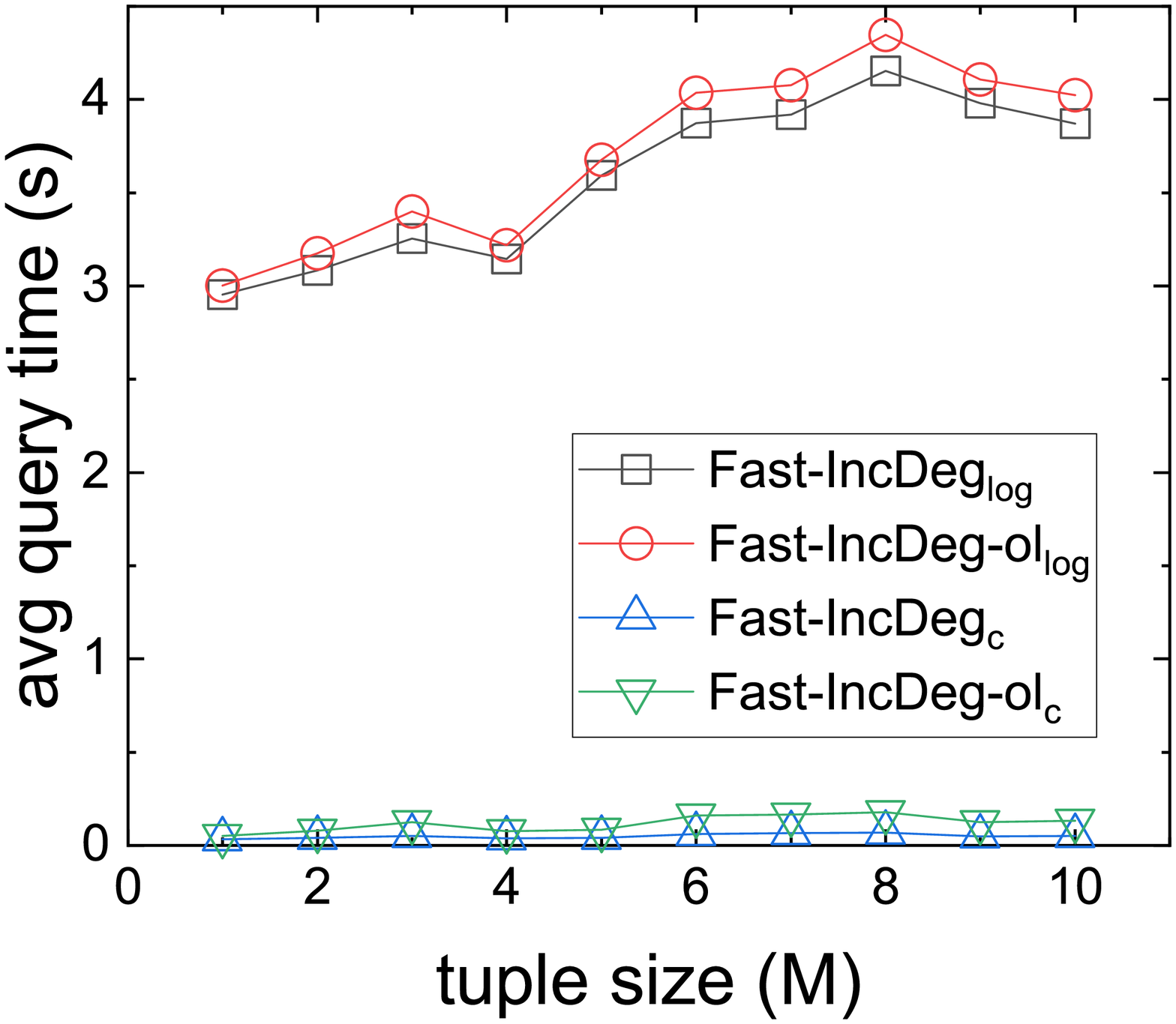} \hspace{-3mm}
			\label{fig:fi-ord-drho}
		}
		\subfigure[\textsc{order} $n = 10\mathrm{M}, \rho = 0.03$]{
			\hspace{-3mm}	\includegraphics[height=35mm]{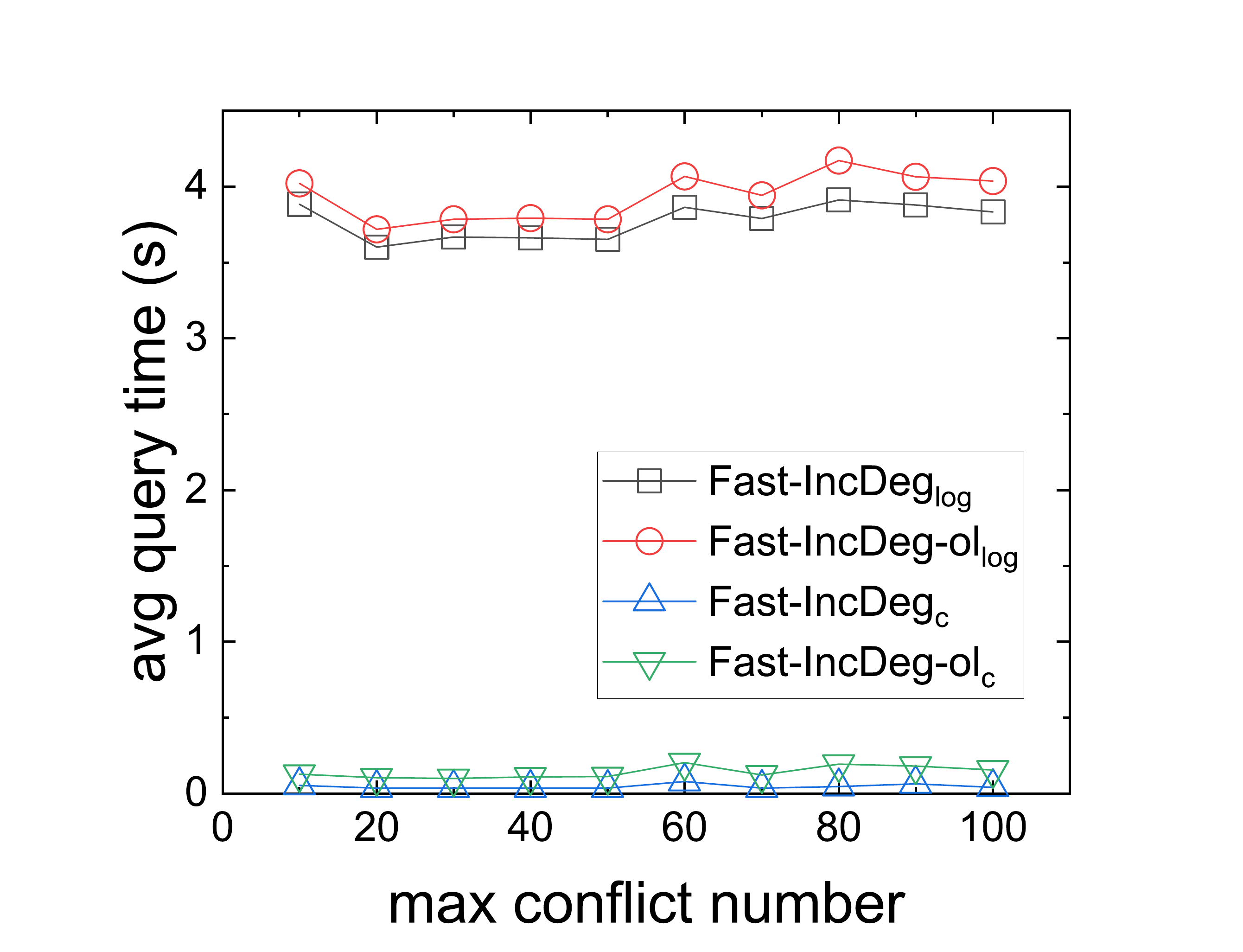} \hspace{-3mm}
			\label{fig:fi-ord-nrho}
		}
		\vspace{-4mm}
		
		\subfigure[\textsc{dblp} $\rho  = 0.05$]{
			\hspace{-3mm}\includegraphics[height=35mm]{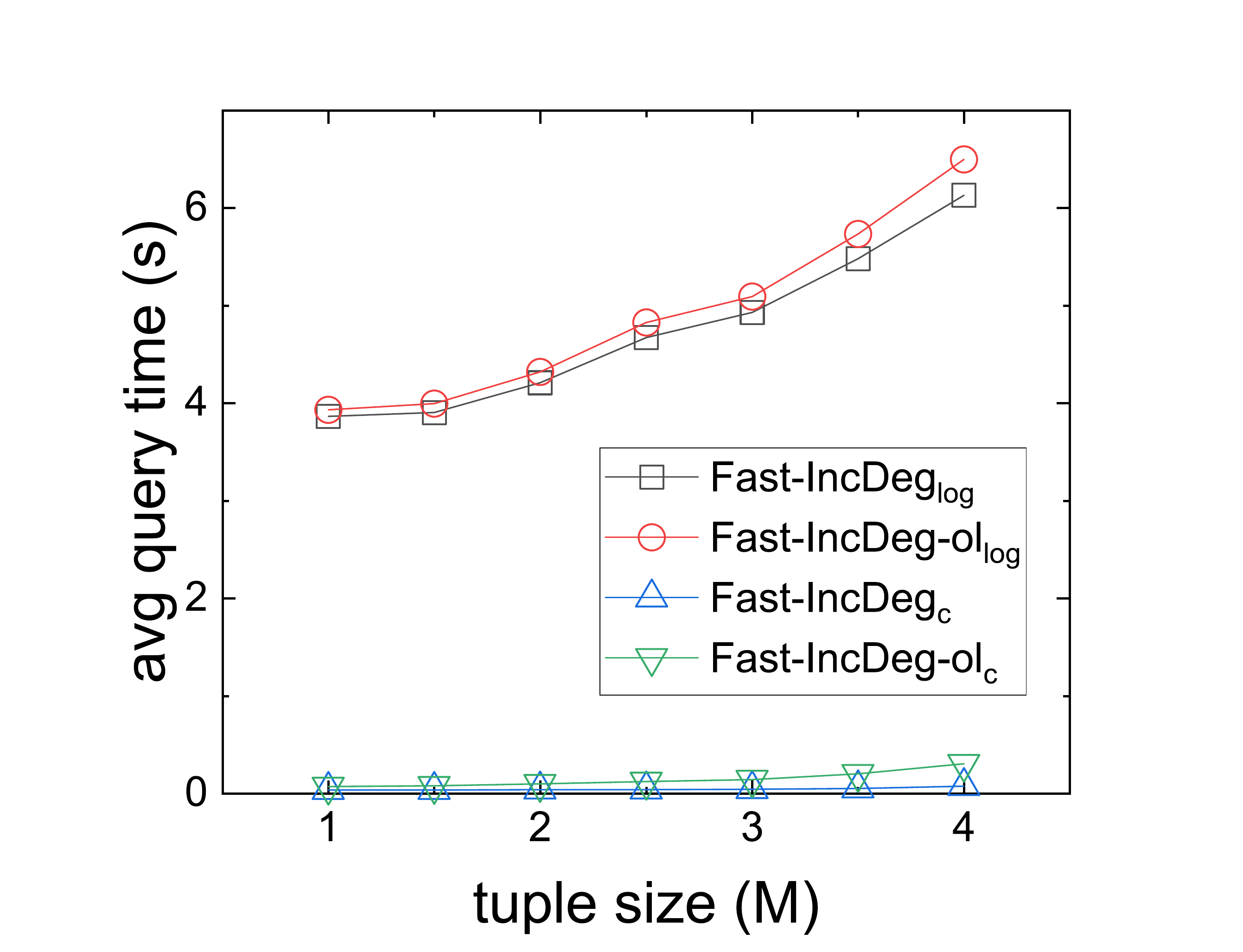} \hspace{-3mm}
			\label{fig:fi-dblp-rho}
		}
		\subfigure[\textsc{dblp} $n = 4$M]{
			\hspace{-3mm}	\includegraphics[height=35mm]{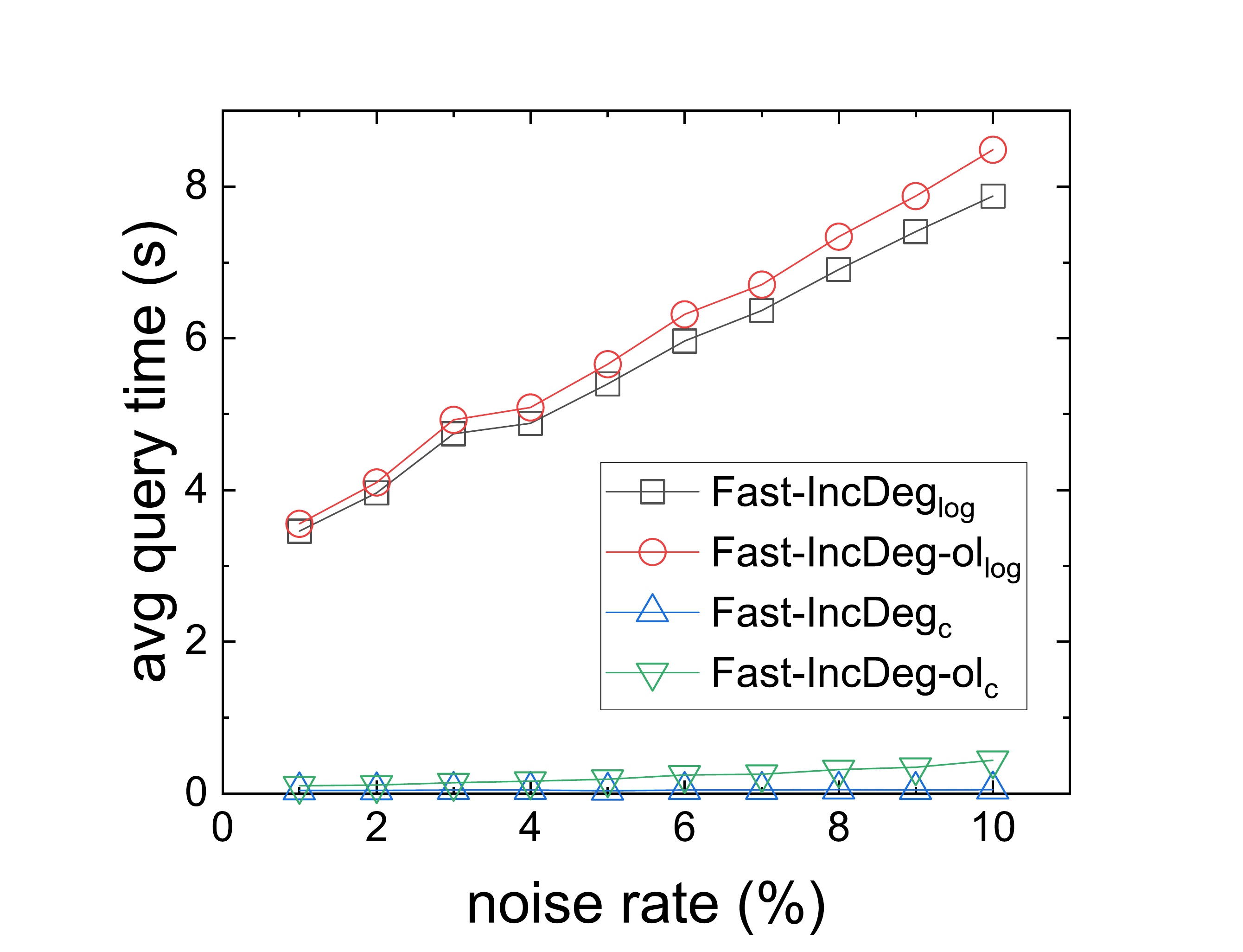} \hspace{-3mm}
			\label{fig:fi-dblp-n}
		}
		\subfigure[\textsc{dblp} $d= 50, \rho = 0.03$]{
			\hspace{-3mm}	\includegraphics[height=35mm]{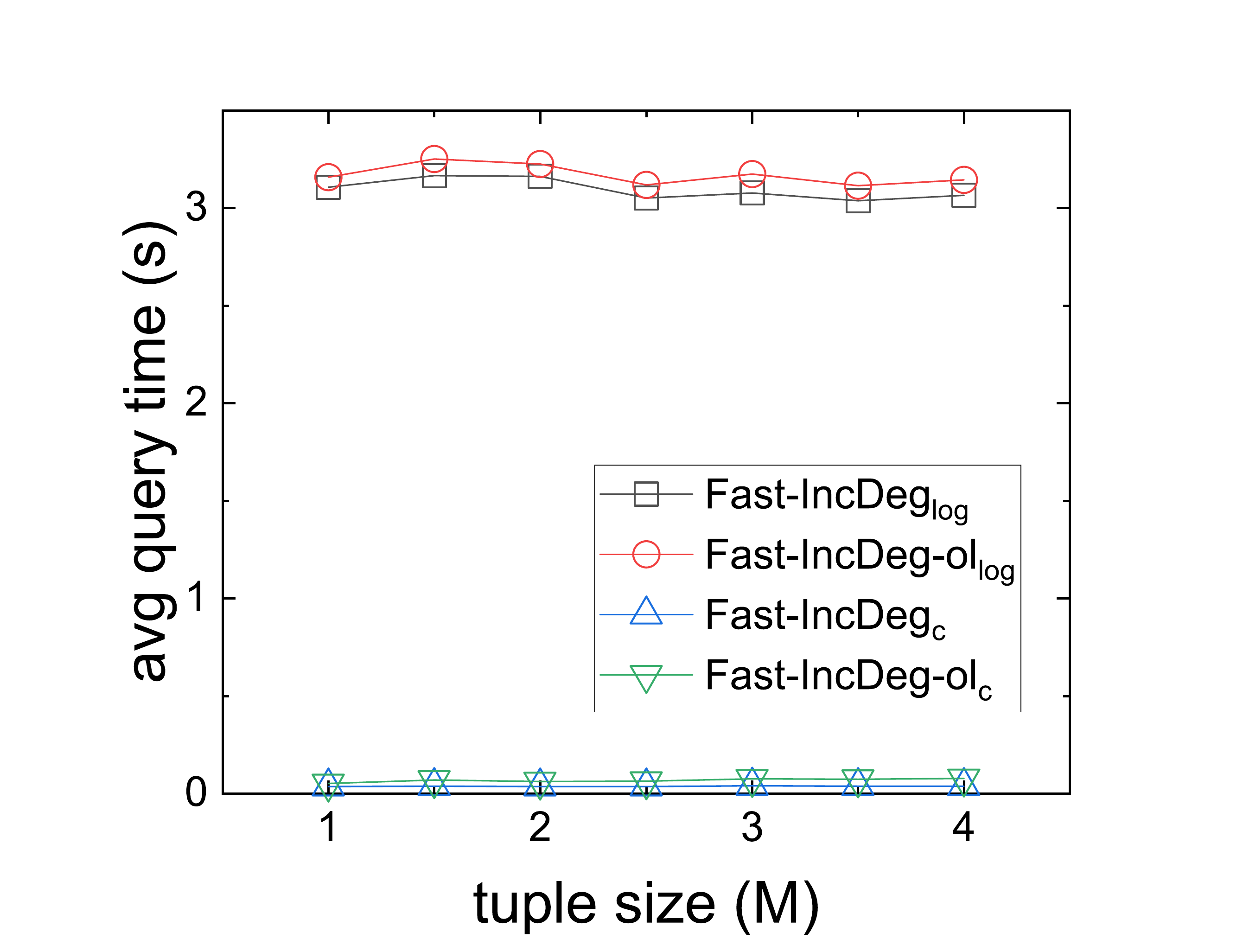} \hspace{-3mm}
			\label{fig:fi-dblp-drho}
		}
		\subfigure[\textsc{dblp} $n = 4\mathrm{M}, \rho = 0.03$]{
			\hspace{-3mm}	\includegraphics[height=35mm]{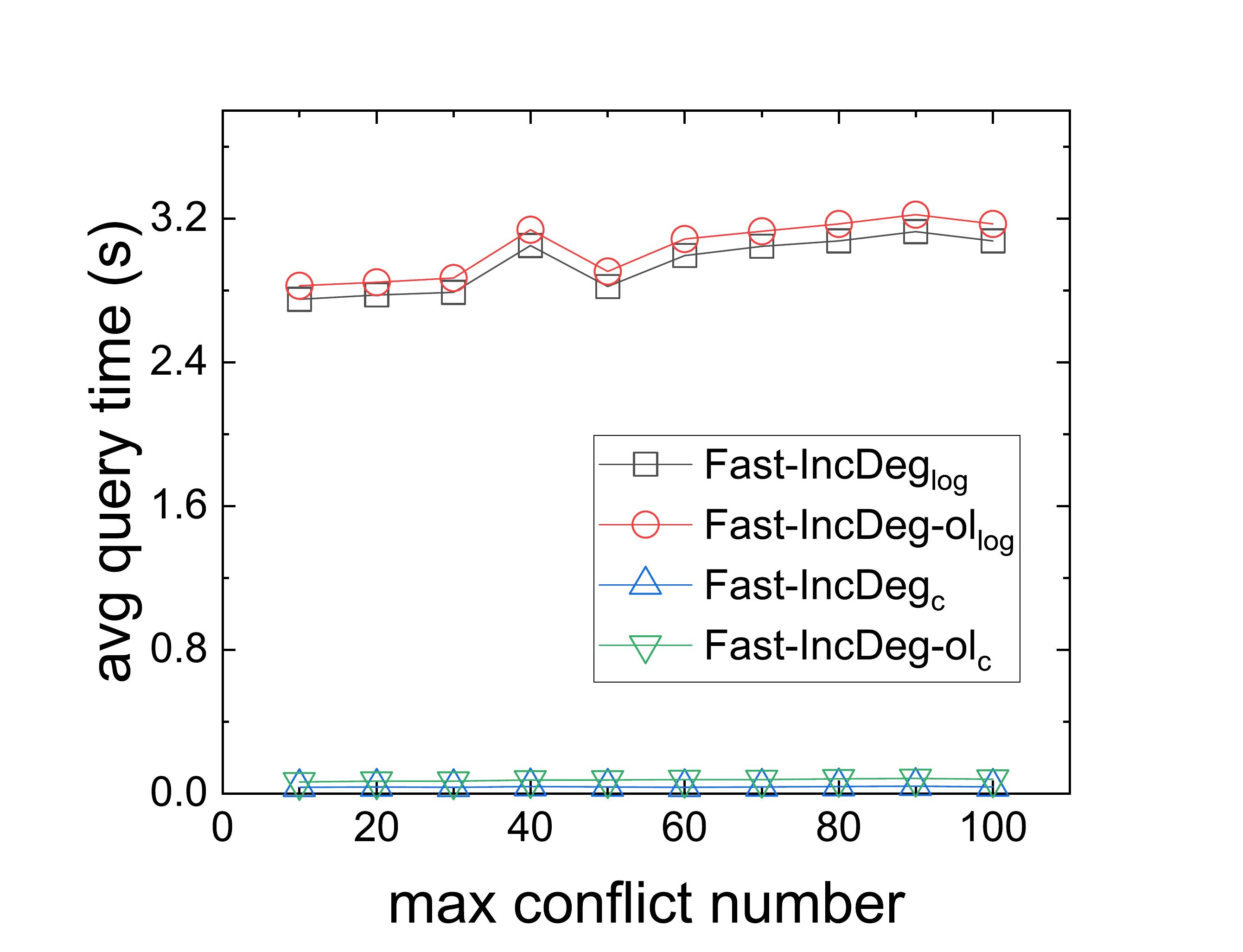} \hspace{-3mm}
			\label{fig:fi-dblp-nrho}
		}
		\vspace{-3mm}
		\caption{average query time with different $n$, $\rho$ and $d$}
		\label{fig:fi-exp}
		\vspace{0mm}
	\end{small}
\end{figure*}

\noindent\textbf{Accuracy.}
We first show the accuracy of \textsf{BL LP-OSR}, \textsf{TE LP-OSR} and \textsf{QT LP-OSR}. In figure \ref{fig:lp-exp}, we ran them on datasets consisting of 10K to 40K tuples with noise rate $\rho$ ranging from $1\%$ to $10\%$ and calculated $\mathsf{UB_1} = 2\rho n$. The $dist_{\texttt{sub}}(\hat{J}, I)$ of \textsf{BL LP-OSR}, \textsf{TE LP-OSR} and \textsf{QT LP-OSR} are much less than $\mathsf{UB_1}$. Because the approximation ratio only bounds the relation between worst case output of an algorithm and the optimal solution, \textsf{BL LP-OSR} sometimes performs better than the other two improved algorithm. What's more, it is discovered that, after triad elimination, the ratio of 0.5 solution becomes much less since they only appear when some conflicts form a cycle with odd length.

We also evaluate the accuracy of \textsf{Fast-IncDeg}. We ran $\mathsf{\textsf{Fast-IncDeg}_{c}}$ with parameter $\epsilon = 0.01$ on the same datasets and calculated $\mathsf{UB_2} = (2 \rho + \epsilon) n$. As shown in figure \ref{fig:lp-exp} the value return by $\mathsf{\textsf{Fast-IncDeg}_{c}}$ is less than $\mathsf{UB_2}$ even less than $\mathsf{UB_1}$ mostly since the upper bound is loose. And it is greater than the values of \textsf{BL LP-OSR}, \textsf{TE LP-OSR} and \textsf{QT LP-OSR} since it returns an estimate with an additive error. \smallskip

\noindent\textbf{Efficiency.}
We evaluate the efficiency of our algorithms for FD-inconsistency evaluation. We first ran $\mathsf{\textsf{Fast-IncDeg}_{c}}$, $\mathsf{\textsf{Fast-IncDeg}_{log}}$, $\mathsf{\textsf{Fast-IncDeg-ol}_{c}}$ and $\mathsf{\textsf{Fast-IncDeg-ol}_{log}}$ on datasets with various size, noise rate $\rho$ ranging from $1\%$ to $10\%$ and $\epsilon = 0.01$. 300 large queries were issued per dataset and the average query time was recorded. 

Figure \ref{fig:fi-ord-rho}, \ref{fig:fi-ord-n}, \ref{fig:fi-dblp-rho}, and \ref{fig:fi-dblp-n} indicate that the average query time increase with the number of tuples and noise rate since both of them influence the maximum conflicts number in $Q(I)$. Further experiments were performed to evaluate the impact of the maximum conflict number $\delta_{Q(I)}$ on the average query time. Since queries were generated randomly, we only bounded the maximum conflict number of the dataset $\delta_D$. Therefore, the average query time shown in figures \ref{fig:fi-ord-nrho} and \ref{fig:fi-dblp-nrho} remain basically the same with the increasing $\delta_D$. As shown in figures \ref{fig:fi-ord-drho} and \ref{fig:fi-dblp-drho}, the average query time of $\mathsf{\textsf{Fast-IncDeg}_{c}}$ and $\mathsf{\textsf{Fast-IncDeg-ol}_{c}}$ change slightly due to the impact of the number of tuples on $\delta_{Q(I)}$. And the average query time of $\mathsf{\textsf{Fast-IncDeg}_{log}}$ and $\mathsf{\textsf{Fast-IncDeg-ol}_{log}}$ grow with the number of tuples. 

Figure \ref{fig:fi-exp} also illustrates that no matter how the $\subseteq$-oracle is implemented, \textsf{Fast-IncDeg} performs better then \textsf{Fast-IncDeg-ol}. It is because that in \textsf{Fast-IncDeg-ol} the ranking is assigned on-the-fly when a conflict is queried and it is expensive to keep the ranking consistent in every tuple which the conflict concerned about. In addition, an efficient $\subseteq$-oracle indeed makes the average query time drop a lot.

\eat{Figure \ref{fig:ord_eff_rho}, \ref{fig:ord_eff_n}, \ref{fig:dblp_eff_rho} and \ref{fig:dblp_eff_n} indicate that the average query increases with the tuples size and noise rate since both of them will influence the maximum conflicts number in $Q(I)$.
A further is performed to evaluate the impact of the maximum conflict number $\delta_Q(I)$ on the average query time.
Since queries were generated randomly, we only bounded the maximum conflict number of the dataset $\delta_D$.
As shown in figures \ref{fig:ord_eff_drho} and \ref{fig:dblp_eff_drho}, the average query time of $\mathsf{\textsf{Fast-IncDeg}_{c}}$ and $\mathsf{\textsf{Fast-IncDeg}_{c}\_ol}$ grows slowly with the number of tuples since $\delta_Q(I)$ also increases with the tuples size. And the average query time of $\mathsf{\textsf{Fast-IncDeg}_{log}}$ and $\mathsf{\textsf{Fast-IncDeg}_{log}\_ol}$ increases logarithmically. For the same reason the average query time in \ref{fig:ord_eff_nrho} and \ref{fig:dblp_eff_nrho} the average query time remains basically the same with increasing $\delta_D$.

Figure \ref{fig:eff} also illustrates that no matter how the $\subseteq$-oracle is implemented, \textsf{Fast-IncDeg} performs better than \textsf{Fast-IncDeg\_ol}. It is because that in \textsf{Fast-IncDeg\_ol} the ranking is assigned on-the-fly when a conflict is queried and it is expensive to keep the ranking consistent in every tuple which the conflict concerned about. In addition, an efficient $\subseteq$-oracle indeed makes the average query time drop a lot.}

\section{Related Work}
As a principled approach managing inconsistency, Arenas et al.~\cite{arenas99cqa} introduced the notions of repairs to define consistent query answering. 
The definitions of repair differ in settings of integrity constraints and operation gain~\cite{Afrati2009}.
The most general form of integrity constraints are denial constraints~\cite{Gaasterland1992}, they are able to express the classic functional dependencies~\cite{Abiteboul1995Foundations}, inclusion dependencies~\cite{Koehler2017}, and so on.
Data complexities of computing optimal repairs are widely studied in the past.
The complexity of tuple-level deletion based subset repair~\cite{Chomicki2005MIM,ester:PODS} is studied respectively in the past.
And the complexity of cell-level update based v-repair is also studied in~\cite{ester:PODS,Kolahi2009AOR}.
APXcompleteness of both optimal subset repair and v-repair computation has been shown for in these works.

For the upper bound, the best approximation on subset repair is still 2 obtained by solving the corresponding vertex cover problem~\cite{ester:PODS} without the limitation on the number of given FDs.
For the setting of fixed number of FDs, there are still no existing algorithmic result.

For the data repairing frameworks~\cite{Afrati2009RCI}, there are two kinds of works which are based on FDs,
they both aim to directly resolve the inconsistency of database.
One kind of methods is to repair data based on minimizing the repair cost, \eg,~\cite{arenas99cqa, bohannon2005cost, cong2007improving, lopatenko2007efficient, winkler2004methods}.

Given the data edit operations (including tuple-level and cell-level),
minimum cost repair will output repaired data with minimizing the difference between it and the original one.
But these work also do not provide us tight lower and upper bounds for data repairing.
There are some other type of repairs not related with this paper, such as “minimum description length”~\cite{Chiang11UnifiedModel}, “relative trust”~\cite{beskales2012relative} and so on.

For inconsistency detection, there exists some detection techniques which are able to detect errors efficiently.
SQL techniques for detecting FD violations were given in~\cite{Chomicki2005MIM},
practical algorithms for detecting violations of FDs in fragmented and distributed relations were provided in~\cite{fan2010detecting},
and a incremental detection algorithm were proposed by~\cite{Fan2014IDI}.
In contrast to inconsistency detection,
inconsistency evaluation need to compute the quantized dirtiness value of the data,
rather than finding all violations.
\section{Conclusions}
We revisit computing an optimal s-repair problem and fast estimate of s-repair based FD-inconsistency degree of subset query results.
For the lower bound, we improve the inapproximability of optimal s-repair computing problem over most cases of FDs and schemas.
For the upper bound, we developed two LP-based algorithms to compute a near optimal s-repair based on different characterization of input FDs and schemas respectively.
Complexity results implies it is hard to obtain a good approximation polynomially, not to mention sublinear time for large data.
For the FD-inconsistency degree, we present a fast $(2,\epsilon)$-estimation with an average sublinear query complexity, and achieve a sublinear time complexity whenever incorporating a sublinear implementation of the subset query oracle.
This results give a way to estimate FD-inconsistency degree efficiently with theoretical guarantee.
\balance	
\bibliographystyle{abbrv}
\bibliography{vldb_2019_osr.bib}
\pagebreak
%
%
%

\end{document}